%% file: dm-monopoly-final.tex
\documentclass[letterpaper,11pt]{article}
\usepackage[T1]{fontenc}
\usepackage[margin=1in]{geometry}
\usepackage{amsmath,amsfonts,amsthm}
\usepackage{thm-restate}
\usepackage{pifont}
\usepackage{soul}
\usepackage{graphicx}
\usepackage{enumerate}
\usepackage{bm}
\usepackage{verbatim}
\usepackage{diagbox}
\usepackage[hypertexnames=false,bookmarksnumbered=true,final]{hyperref}
\usepackage[capitalize,nameinlink]{cleveref}
\usepackage[dvipsnames]{xcolor}
\usepackage{appendix}
\usepackage{makecell}
\usepackage{enumitem}
\usepackage{tablefootnote}
\usepackage{tikz}
\usetikzlibrary{datavisualization}
\usepackage{xifthen}
\usepackage{subcaption}
\usepackage{hhline}
\usepackage{authblk}

\def\colorschemesepia{sepia}
\def\colorschemedark{dark}
\def\colorschemelight{light}

\ifx\colorscheme\undefined
\let\colorscheme\colorschemelight
\fi

\ifx\colorscheme\colorschemelight
\colorlet{textColor}{black}
\colorlet{bgColor}{white}
\fi

\ifx\colorscheme\colorschemesepia
\definecolor{textColor}{HTML}{433423}
\definecolor{bgColor}{HTML}{fbf0da}
\fi

\ifx\colorscheme\colorschemedark
\definecolor{textColor}{HTML}{bdc1c6}
\definecolor{bgColor}{HTML}{202124}
\definecolor{textRed}{HTML}{ff968c}  %
\definecolor{textGreen}{HTML}{70cc70}  %
\definecolor{textBlue}{HTML}{8cbcff}  %
\definecolor{textCyan}{HTML}{70cccc}  %
\definecolor{textMagenta}{HTML}{d982d9}  %
\definecolor{textYellow}{HTML}{bfbf69}  %
\definecolor{textPurple}{HTML}{c58af9}
\colorlet{textHeavy}{white}
\else
\colorlet{textRed}{red!50!black}
\colorlet{textGreen}{green!35!black}
\colorlet{textBlue}{blue!40!black}
\definecolor{textPurple}{HTML}{995272}
\colorlet{textHeavy}{textColor}
\fi

\ifx\colorscheme\colorschemelight\else
\pagecolor{bgColor}
\color{textColor}
\fi

\colorlet{dimColor}{textColor!50!bgColor}

\hypersetup{
    colorlinks,
    citecolor=textGreen,
    linkcolor=textBlue,
    urlcolor=textPurple,
    pdfpagemode=UseNone,
    pdfstartview=FitW
}

\renewcommand{\v}{\tau}

\let\eps\varepsilon
\newcommand*{\defeq}{:=}
\newcommand*{\Th}{^{\textrm{th}}}

\newcommand*{\wLoG}{without loss of generality}

\newcommand*{\boolOne}{\mathbbold{1}}  %
\newcommand*{\vecOne}{\mathbf{1}}
\newcommand*{\vecZero}{\mathbf{0}}
\newcommand*{\vecE}{\mathbf{e}}

\newcommand*{\linf}[1]{\lVert #1 \rVert_{\infty}}

\DeclareMathOperator*{\E}{\mathbb{E}}
\DeclareMathOperator*{\Var}{Var}

\DeclareMathOperator*{\argmax}{argmax}

\DeclareMathOperator{\conv}{conv}
\DeclareMathOperator{\epi}{epi}
\DeclareMathOperator{\disc}{disc}
\DeclareMathOperator{\plin}{plin}

\DeclareMathOperator{\Dem}{\mathit{OPT}}
\DeclareMathOperator{\Pre}{Pre}
\DeclareMathOperator{\Beta}{Beta}
\DeclareMathOperator{\Unif}{Unif}

\newcommand*{\train}{\mathrm{train}}
\newcommand*{\test}{\mathrm{test}}

\newcommand*{\p}{\bm{p}}
\newcommand*{\x}{\bm{x}}
\newcommand*{\y}{\bm{y}}
\newcommand*{\z}{\bm{z}}
\newcommand*{\bvec}{\bm{b}}
\newcommand*{\pvec}{\bm{p}}
\newcommand*{\rvec}{\bm{r}}
\newcommand*{\xvec}{\bm{x}}
\newcommand*{\yvec}{\bm{y}}
\newcommand*{\zvec}{\bm{z}}
\newcommand*{\alphavec}{\bm{\alpha}}

\DeclareMathAlphabet{\mathbbold}{U}{bbold}{m}{n}

\newcommand*{\D}{\mathcal{D}}
\newcommand*{\I}{\mathcal{I}}
\renewcommand{\P}{\mathcal{P}}
\newcommand*{\Dcal}{\mathcal{D}}

\newcommand*{\ptild}{\widetilde{p}}

\newcommand*{\xtild}{\widetilde{x}}
\newcommand*{\ytild}{\widetilde{y}}
\newcommand*{\ztild}{\widetilde{z}}
\newcommand*{\pvectild}{\widetilde{\pvec}}

\newcommand*{\fhat}{\widehat{f}}
\newcommand*{\ghat}{\widehat{g}}
\newcommand*{\hhat}{\widehat{h}}

\newcommand*{\phat}{\widehat{p}}

\newcommand*{\xhat}{\widehat{x}}
\newcommand*{\Xhat}{\widehat{X}}
\newcommand*{\yhat}{\widehat{y}}
\newcommand*{\Yhat}{\widehat{Y}}
\newcommand*{\zhat}{\widehat{z}}
\newcommand*{\pvechat}{\widehat{\pvec}}
\newcommand*{\xvechat}{\widehat{\xvec}}
\newcommand*{\yvechat}{\widehat{\yvec}}
\newcommand*{\deltahat}{\widehat{\delta}}

\theoremstyle{plain}
\newtheorem{theorem}{Theorem}
\newtheorem*{theorem*}{Theorem}

\newtheorem{lemma}{Lemma}
\newtheorem{observation}{Observation}
\newtheorem{remark}{Remark}
\newtheorem{conjecture}{Conjecture}
\theoremstyle{definition}
\newtheorem{definition}{Definition}
\newtheorem{example}{Example}

\crefname{appsec}{Appendix}{Appendices}
\crefname{lemma}{Lemma}{Lemmas}
\crefname{theorem}{Theorem}{Theorems}
\crefname{definition}{Definition}{Definitions}
\crefname{fact}{Fact}{Facts}
\crefname{claim}{Claim}{Claims}
\crefname{proposition}{Proposition}{Propositions}

\newcommand*{\optprog}[3]{
\begin{array}{*4{>{\displaystyle}l}}
#1 & \multicolumn{3}{>{\displaystyle}l}{#2}
#3 \end{array}}

\tikzset{
  my grid style/.style={
    tikz/data visualization/all axes={
      grid,grid={
        minor={style={draw={textColor!10!bgColor}}},
        major={style={draw={textColor!10!bgColor}}},
      }
    }
  }
}

\allowdisplaybreaks
\makeatletter
\g@addto@macro{\UrlBreaks}{%
\do\/%
\do\a\do\b\do\c\do\d\do\e\do\f\do\g\do\h\do\i\do\j\do\k\do\l\do\m%
\do\n\do\o\do\p\do\q\do\r\do\s\do\t\do\u\do\v\do\w\do\x\do\y\do\z%
\do\A\do\B\do\C\do\D\do\E\do\F\do\G\do\H\do\I\do\J\do\K\do\L\do\M%
\do\N\do\O\do\P\do\Q\do\R\do\S\do\T\do\U\do\V\do\W\do\X\do\Y\do\Z%
\do\0\do\1\do\2\do\3\do\4\do\5\do\6\do\7\do\8\do\9%
}
\makeatother

\makeatletter
\@ifpackageloaded{enumitem}{%
\newenvironment*{tightemize}{\begin{itemize}[noitemsep]}{\end{itemize}}%
\newenvironment*{tightenum}{\begin{enumerate}[noitemsep]}{\end{enumerate}}%
}{%
\newenvironment*{tightemize}{\begin{itemize}}{\end{itemize}}%
\newenvironment*{tightenum}{\begin{enumerate}}{\end{enumerate}}%
}
\makeatother

\let\citep\cite
\let\citet\cite

\title{Non-Linear Pricing Restores Tractability for a Data Seller%
\thanks{J.~Garg and E.~Sharma were supported by NSF grant CCF-2334461.
B.~R.~Chaudhury and J.~Song were supported by the NSF Career Award CCF-2441580.}}

\author[1]{Bhaskar Ray Chaudhury}
\author[1]{Jugal Garg}
\author[1]{Eklavya Sharma}
\author[1]{Jiaxin Song}

\affil[1]{University of Illinois, Urbana-Champaign}
\affil[ ]{\{braycha, jugal, eklavya2, jiaxins8\}@illinois.edu}
\date{\empty}

\begin{document}

\maketitle

\begin{abstract}
We consider a data seller who designs pricing mechanisms over multiple datasets to maximize revenue from budget-constrained buyers.  The seller offers multiple datasets and assigns each a pricing function that maps the quantity purchased to a total payment. The goal is to design these pricing functions to maximize revenue, anticipating that buyers—who trade off accuracy gains against cost—choose bundles optimally subject to their budget constraints.

Prior work \texorpdfstring{\cite{chaudhury2026revenue}}{[Chaudhury et al., 2026]} studies such optimal pricing under the restriction that each dataset is assigned a linear price, and shows that computing optimal linear prices is computationally intractable. In contrast, we allow each dataset to be priced via a general function and show that this additional flexibility can not only increase the revenue but also restore tractability, yielding a surprising simultaneous improvement in economic performance and computational efficiency.

Even when pricing functions are only required to be monotone and lower-continuous, optimal pricing admits a highly structured and simple form: each pricing function is piecewise linear and convex (PLC), and the optimal solution can be computed in polynomial time. Moreover, the total number of \emph{kinks} across all pricing functions is bounded by the number of buyers. Consequently, when datasets significantly outnumber buyers, most pricing functions are effectively linear. We further empirically study the structure of optimal pricing by analyzing the number of kinks and the revenue gap between optimal nonlinear pricing and optimal linear pricing on simulations generated from a real dataset.

\end{abstract}

\section{Introduction}
\label{sec:intro}

The rise of machine learning and artificial intelligence has transformed data into one of the most valuable economic resources of our time. As these technologies have been widely adopted across industry and society over the past decade, the demand for large-scale, high-quality datasets has grown dramatically. The global big data market is projected to reach \$473.6 billion by 2030 \citep{AcumenDM}. Given the growing economic importance of data, developing a rigorous theory of data pricing is a crucial step toward fully realizing its value.

In this paper, we study a setting in which a monopolistic data seller offers $m$ datasets to $n$ budget-constrained buyers. The seller sets a price for each dataset, after which each buyer selects a utility-maximizing subset subject to their budget constraint. Our goal is to address the question ``How should each dataset be priced to maximize the seller's total revenue?"

This problem was recently studied by \citet{chaudhury2026revenue}, who consider a model in which buyers can purchase fractional amounts of datasets. Specifically, each dataset comprises a large number of records, and buyers may choose to purchase any fraction of them. The authors assume \emph{linear pricing}: if a dataset is priced at $p \in \mathbb{R}_{\ge 0}$ and a buyer purchases an $x \in [0, 1]$ fraction, then the payment is $p \cdot x$. They show that, given buyers' budgets and utilities for the datasets, computing revenue-maximizing prices is APX-hard. %

In this work, we consider more general pricing schemes beyond linear pricing. For each dataset, we specify a \emph{pricing function} $p : [0, 1] \rightarrow \mathbb{R}_{\geq 0}$ that maps the purchased fraction to a non-negative payment. That is, if a buyer purchases an $x \in [0,1]$ fraction of a dataset, she pays $p(x)$. Apart from a natural requirements that $p$ be monotone%
\footnote{monotonicity means that purchasing a larger fraction of a dataset cannot cost less.} and lower continuous, we impose no additional structural assumptions on the pricing function.

There are two main motivations for considering more general pricing schemes. First, expanding the space of allowable pricing functions can lead to higher revenue. In particular, we provide examples showing that moving from linear to general pricing can increase revenue by up to a factor of two. Second, and perhaps more surprisingly, \emph{the revenue-maximization problem becomes tractable under general pricing!}

Our approach proceeds in two steps. We first derive a structural result showing that there always exists a revenue-maximizing solution within the class of piecewise linear convex (PLC) pricing functions. We then formulate a linear program that computes the optimal PLC pricing functions.

To formally state our model and results, we first introduce the buyers' utility model over the datasets and the necessary notation.

\subsection{Our Model}
\label{sec:our-model}

We consider a set $N$ of $n$ buyers and a set $M$ of $m$ datasets. Each buyer $i \in N$ has a budget $b_i$.
Buyers aim to improve the accuracy of their machine learning models' predictions. Accordingly, data derives its value from its ability to improve these predictions.

Specifically, each buyer $i \in N$ wants to estimate an unknown parameter $\theta_i \in \mathbb{R}$, whose prior distribution is known to her. Her \emph{data bundle} $\xvec_i \defeq (x_{i,1}, x_{i,2}, \ldots, x_{i,m})$%
---the collection containing an $x_{i,j}$ fraction of each dataset $j \in M$---%
allows her to update her belief about $\theta_i$.
A buyer's value for bundle $\xvec_i$ is proportional to the \emph{accuracy gain}, defined as the reduction in uncertainty about~$\theta_i$:
\[ a_i(\xvec_i) \defeq \E\!\left[\Pre(\theta_i \mid S(\xvec_i))\right] - \Pre(\theta_i), \]

where $\Pre(\theta) \defeq 1 / \Var(\theta)$.
This formulation follows the framework of \citet{ChaudhuryGMS26,chaudhury2026revenue,ChaudhuryGSS2026},
who build on the literature on the \emph{value of data} \citet{baley2025data}.
Under Gaussian priors with additive Gaussian noise, it can be shown \citep{chaudhury2026revenue} that precision increases linearly in the data bundle. Consequently, the accuracy gain takes the form
\[ a_i(\xvec_i) = \sum_{j=1}^m \tau_{i,j}x_{i,j}. \]
Here $\tau_{i,j}$ can be interpreted as a measure of the informativeness of dataset $j$ for buyer $i$.

As is standard in revenue maximization and auction theory~\citep{myerson1981optimal,klemperer2004auctions}, we assume that a buyer's \emph{net utility} from a data bundle~$\xvec_i$ equals her value for the accuracy gain minus the total payment.
Let $p_j(\cdot)$ denote the pricing function for the $j\Th$ dataset, and define the aggregate pricing function $\pvec(\cdot)$ as $\pvec(\xvec_i) = \sum_{j=1}^m p_j(x_{i,j})$. Then the net utility of buyer $i$ from bundle~$\xvec_i$ is given by
\[ u_i(\xvec_i, \pvec) = \alpha_i \cdot a_i(\xvec_i) - \sum_{j=1}^m p_j(x_{i,j}), \]
where $\alpha_i$ captures buyer $i$'s valuation for accuracy improvements. Since $a_i$ is linear, we can assume \wLoG{} that $\alpha_i = 1$ for all $i\in N$.

\noindent\textbf{Buyers demand optimal bundles.}
Once the pricing functions are fixed, each buyer selects a data bundle that maximizes her utility subject to her budget constraint. Formally, given the aggregate pricing function $\pvec$, buyer~$i$'s \emph{demand set} $\Dem_i(\pvec)$ is defined as
\begin{align*}
\Dem_i(\pvec) \defeq \argmax_{\yvec \in [0, 1]^m:\,\pvec(\yvec) \le b_i} u_i(\yvec, \pvec).
\end{align*}
Here, $\yvec = (y_1, \dots, y_m)$ denotes the fractions of datasets purchased, and the constraint $\pvec(\yvec) \le b_i$ ensures that the total cost does not exceed the buyer's budget.

\noindent\textbf{The revenue-maximization problem.}
The seller aims to choose pricing functions that maximize total revenue, assuming that each buyer purchases a utility-maximizing bundle.
We consider a broad class of pricing functions: those that are monotone---meaning that purchasing a large fraction of a dataset cannot cost less---and lower-continuous. As illustrated in \cref{ex:no-util-max} (\cref{sec:rev_max}), if pricing functions fail to satisfy lower continuity, a utility-maximizing demand bundle for a buyer may fail to exist. Thus, lower continuity is essential for ensuring that our problem is well-defined. This class of functions is sufficiently rich to capture a wide range of economically meaningful pricing schemes.

Under these assumptions, the resulting bi-level optimization problem can be expressed as follows:
\begin{equation}
\label{pgm:rev-max}
\begin{aligned}
\max_{\pvec,\; \xvec_i \in \Dem_i(\pvec) \,\forall i}
\quad &
\sum_{i=1}^n \pvec(\xvec_i) = \sum_{i=1}^n \sum_{j=1}^m p_j(x_{i,j})
\\ \text{s.t.} \quad & p_j(\cdot) \text{ is monotone lower-continuous, for all } j \in M.
\end{aligned}
\end{equation}

We remark that general pricing functions are particularly well-suited to non-rivalrous assets such as data. In contrast, in the traditional setting of rivalrous goods, such pricing functions are less desirable, as a buyer could potentially exploit an increasing pricing scheme by duplicating into multiple identities to purchase smaller quantities at lower marginal prices. However, for non-rivalrous assets such as data, this behavior is self-defeating: duplicate buyers would obtain overlapping, redundant data records that do not further reduce uncertainty or increase utility. Overall, non-linear pricing thus emerges as a natural and practically relevant approach for selling datasets.

\subsection{Our Contributions}
\label{sec:our-contrib}

Our main contribution is to prove that in the monopolistic data pricing problem, revenue-maximizing pricing functions can be found in polynomial time, i.e., the optimization problem \eqref{pgm:rev-max} can be solved in polynomial time. To establish this result, we make several key observations along the way.

\begin{enumerate}
\item \textbf{Existence of structured optima}:
    We first show the existence of revenue-maximizing pricing functions satisfying certain desirable properties
    (see \cref{thm:struct} in \cref{sec:overview}).
    Specifically, we show that for any aggregate pricing function $\pvec = (p_1, \ldots, p_m)$,
    there exists another aggregate pricing function $\pvechat = (\phat_1, \ldots, \phat_m)$
    whose revenue is at least as large as that of $\pvec$,
    and such that for every dataset $j \in [m]$, the function $\phat_j$ is piecewise-linear and convex (PLC).
    Moreover, for each $x \in (0, 1]$ the left derivative of $\phat_j$ at $x$ belongs to the set $\{\v_{i,j}: i \in N\}$.

\item \textbf{LP formulation}:
    By the preceding argument, selecting a pricing function is equivalent to choosing the \emph{lengths} of \emph{segments} of each $\phat_j$ (see \cref{fig:plc-ex} for an example; we further elaborate on this connection in \cref{sec:sharding-lp}). This observation effectively reduces the pricing space from an infinite-dimensional
    space of functions to $m$ vectors in $\mathbb{R}_{\ge 0}^n$. As a result, we can reformulate the revenue-maximization problem as a linear program, which yields a polynomial-time algorithm.

\begin{figure*}[htb]
\centering
\begin{subfigure}[b]{0.45\textwidth}
\centering
\begin{tikzpicture}[scale=0.6,baseline={(current bounding box.center)}]
\datavisualization [
school book axes,
visualize as line={my data},
my grid style,
my data={no lines},
x axis={length={5cm},label=$x$,ticks={step=0.2}},
y axis={length={5cm},label=$\phat_j(x)$,ticks={step=5}},
]
data {
x,y
0,0
1,30
}
info {
\path [draw={textBlue},very thick]
       (visualization cs: x=0.0, y= 0)
    -- (visualization cs: x=0.4, y= 4)
    -- (visualization cs: x=0.8, y=14)
    -- (visualization cs: x=1.0, y=27);
};
\end{tikzpicture}

\caption{Pricing function}
\label{fig:plc-ex:func}
\end{subfigure}
\begin{subfigure}[b]{0.45\textwidth}
\centering
\begin{tikzpicture}[scale=0.6,baseline={(current bounding box.center)}]
\datavisualization [
school book axes,
visualize as line={my data},
my grid style,
my data={no lines},
x axis={length={5cm},label=$x$,ticks={step=0.2}},
y axis={length={5cm},label=$\phat_j'(x)$,ticks={step=10}},
]
data {
x,y
0,0
1,70
}
info {
\path [draw={textBlue},very thick]
    (visualization cs: x=0.0, y=10)
    -- (visualization cs: x=0.4, y=10)
    (visualization cs: x=0.4, y=25)
    -- (visualization cs: x=0.8, y=25)
    (visualization cs: x=0.8, y=65)
    -- (visualization cs: x=1.0, y=65);
\path [draw={textBlue},very thick,fill={bgColor}]
    (visualization cs: x=0.0, y=10) circle [radius=2pt]
    (visualization cs: x=0.4, y=25) circle [radius=2pt]
    (visualization cs: x=0.8, y=65) circle [radius=2pt];
};
\end{tikzpicture}

\caption{Left derivative}
\label{fig:plc-ex:der}
\end{subfigure}
\caption{A PLC pricing function $\phat_j$ having three \emph{segments}:
$[0, 0.4]$, $[0.4, 0.8]$, and $[0.8, 1]$.}
\label{fig:plc-ex}
\end{figure*}
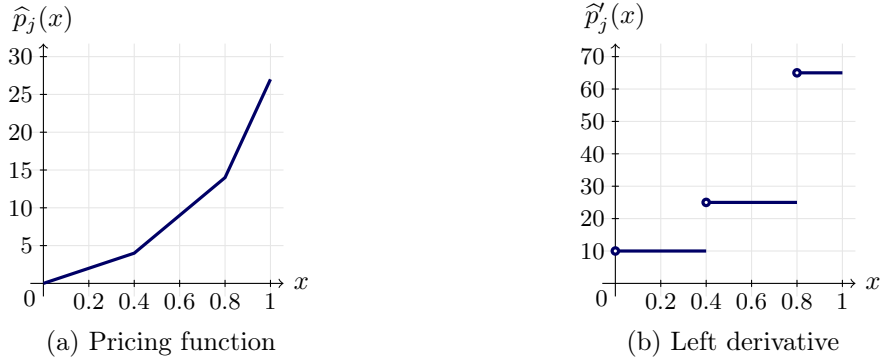

\item \textbf{Sparsity of solutions}:
    Next, using properties of extreme points of the linear program, we show that a sparse optimal solution always exists: the total number of \emph{kinks} (i.e., interior points of non-differentiability) across all pricing functions is at most $n$. Consequently, when the number of datasets significantly exceeds the number of buyers,
    most datasets are priced linearly. In this sense, the revenue-maximizing pricing scheme is \emph{almost} linear.

\item \textbf{Empirical insights}:
    Finally, we empirically investigate properties of revenue-maximizing pricing functions by constructing data pricing instances from a real-world dataset. We define the \emph{linearity gap} of an instance to be the ratio of the maximum revenue achievable under general pricing to that under linear pricing. Our experiments indicate that, in realistic instances, both the linearity gap and the number of kinks are substantially smaller than those observed in specially constructed worst-case examples.
\end{enumerate}

\subsection{Related Work}
\label{sec:related-work}

Our work adopts a very general model for data valuation, defining a buyer's utility in terms of improvements in prediction accuracy. A number of studies have examined more domain-specific models of how agents value data~\citep{farboodi2025valuing, veldkamp2023valuing, farboodi2023data}. We refer the reader to~\citet{fleckenstein2023review} for a detailed review of data valuation methods.
A significant body of work has also studied mechanisms to incentivize data sharing among agents~\citep{fallah2024optimal, cummings2023optimal, fallah2022bridging, MurhekarYCLM23, acemoglu2023good, ghosh2011selling, nissim2014redrawing, chen2018optimal, varian2009economic} typically by compensating participants for their privacy loss. Relatedly, another line of work has examined mechanisms that incentivize sellers to truthfully report the variances of their datasets to a data aggregator tasked with achieving a specified variance level for prediction accuracy~\citep {cummings2015accuracy}.

A large body of work studies the problem of a monopolistic seller offering one or more datasets to buyers
\citep{AgarwalDS19, admati1986monopolistic, admati1990direct, bergemann2018design, BabaioffKP12, mehta2021sell, cai2020sell, bergemann2022economics, horner2016selling, ChawlaRTT22, ChenHK24}. However, these works differ substantially from ours in their modeling assumptions, including what buyers and sellers know about each other, whether information flows between them, whether the seller can price-discriminate, whether datasets can be modified or combined, and whether the buyers are budget-constrained. We refer the reader to \cite{pei2020survey} for a survey on data pricing.

Even when multiple data sellers are present, they often participate through a centralized data marketplace that sets prices and subsequently distributes revenue among sellers \cite{AgarwalDS19}. This effectively reduces the system to a single monopolistic seller, making the monopolistic setting relevant even in multi-seller environments.

There is also a line of work studying equilibrium behavior and auctions in data markets in the presence of externalities~\citep{AgarwalDHR20, Hossain024,ichihashi2021competing}.
In contrast to our monetary setting, prior work has investigated stable solutions in data exchange economies, where agents exchange data without monetary transfers \citep{BhaskaraGIKMS24,akrami2025theoretical,song2025existence}.

\section{Technical Overview}
\label{sec:overview}

We now outline the main ideas underlying our results,
and defer the details to \cref{sec:price-char}.
We begin by reviewing the pricing model and formally defining the revenue of a pricing function.

\subsection{The Model}
\label{sec:overview:model}

Given $m$ datasets, a pricing function is a function $\p: [0, 1]^m \to \mathbb{R}_{\ge 0}$ such that
$\p(\vecZero) = 0$ (i.e., we do not charge a buyer who did not buy anything)
and $\p$ is monotone, i.e., if $x_j \le y_j$ for all $j \in [m]$, then $\p(\xvec) \le \p(\yvec)$,
since it does not make sense to charge more for a smaller bundle.

Note that this model is more general than the one we introduced in \cref{sec:intro}.
The pricing functions there were \emph{separable}, i.e.,
we assumed $\p(\xvec) = \sum_{j=1}^m p_j(x_j)$.
We consider a more general model here because it simplifies some of our proofs.

\begin{definition}[Separable function]
A function $f: [0, 1]^m \to \mathbb{R}$ is called \emph{separable} if
there exist functions $f_1, \ldots, f_m: [0, 1] \to \mathbb{R}$ such that
$f(\xvec) = \sum_{j=1}^m f_j(x_j)$.
The functions $f_1, \ldots, f_m$ are called the \emph{components} of $f$.
\end{definition}

We now define the revenue of a pricing function.

\begin{definition}[Revenue of a pricing function]
\label{defn:revenue}
Consider a pricing function $\p: [0, 1]^m \to \mathbb{R}_{\ge 0}$,
and a buyer having valuation function $\v: [0, 1]^m \to \mathbb{R}_{\ge 0}$
and budget $b \in \mathbb{R}_{\ge 0} \cup \{\infty\}$.
\begin{enumerate}[leftmargin=0.5cm]
\item Let $F(\p \mid b) \defeq \{\xvec \in [0, 1]^m: \p(\xvec) \le b\}$ be the \emph{feasible set}
    or the \emph{affordable set} for $\p$.
    Note that $\vecZero \in F(\p \mid b)$ and $F(\p \mid \infty) = [0, 1]^m$.
\item Let $u^*(\p \mid \v, b) \defeq \sup_{\xvec \in F(\p \mid b)} (\v(\xvec) - \p(\xvec))$
    be the \emph{maximum net utility} for $\p$.
\item Let $\Dem(\p \mid \v, b) \defeq \{\xvec \in F(\p \mid b): \v(\xvec) - \p(\xvec) = u^*(\p \mid \v, b)\}$
    be the \emph{demand set} of $\p$.
\item Let $r(\p \mid \v, b) \defeq \sup_{\xvec \in \Dem(\p \mid \v, b)} \p(\xvec)$
    be the \emph{revenue} of $\p$.
\end{enumerate}
\end{definition}

We remark that the demand set in \cref{defn:revenue} may be empty, i.e.,
there may not exist a utility-maximizing bundle for the buyer.
However, whenever the pricing function is MLC (monotone lower-continuous),
the demand set is always non-empty, so the revenue is well-defined.
We defer the formal definition of MLC to
\cref{sec:real-analysis:mlc} (page~\pageref{sec:real-analysis:mlc}).
See \cref{sec:price-char:mlc} for a discussion on pricing functions being MLC.

\subsection{Structural Theorem}
\label{sec:overview:struct}

The key ingredient in our result is the \emph{structural theorem},
which states that revenue can be maximized using a piecewise-linear convex (PLC) pricing function.

\begin{restatable}[Structural Theorem]{theorem}{thmStruct}
\label{thm:struct}
Consider $m$ datasets and $n$ buyers, where each buyer $i \in [n]$ has
budget $b_i$ and a linear valuation function $\v_i$,
where $\v_i(x) = \sum_{j=1}^m \v_{i,j}x_j$ for all $x \in [0, 1]^m$.
Let $\p: [0, 1]^m \to \mathbb{R}_{\ge 0}$ be a separable MLC pricing function.
Then there exists a separable pricing function $\pvechat$ such that the following hold:
\begin{tightenum}
\item $r(\pvechat \mid \v_i, b_i) \ge r(\p \mid \v_i, b_i)$ for all $i \in [n]$.
\item Each component of $\pvechat$ is monotone, continuous, convex, and piecewise-linear.
\item For each component $\phat_j$ of $\pvechat$, at any point $x \in (0, 1]$,
    the left derivative $\phat_j'(x)$ belongs to the set $\{\v_{i,j}: i \in [n]\}$.
\end{tightenum}
\end{restatable}

To prove this, we begin with a separable MLC pricing function
and show how it can be transformed into a PLC function without any loss of revenue.
We do this transformation in two steps: \emph{convexification} and \emph{piecewise-linearization}.

\subsubsection{Step 1: Convexification}
\label{sec:overview:convexify}

\begin{definition}[Convex envelope of a function \citep{rockafellar1997convex05}]
\label{defn:convex-env}
For a function $f: [0, 1]^m \to \mathbb{R}$, define
\[ U_f(\xvec) \defeq \left\{y \in \mathbb{R}: \begin{array}{ll}
    \exists k \in \mathbb{N}, \exists \lambda \in \Delta_k, \exists \xvec^{(1)}, \ldots, \xvec^{(k)} \in [0, 1]^m
    \\ \text{such that } \xvec = \sum_{i=1}^k \lambda_i \xvec^{(i)} \text{ and } y \ge \sum_{i=1}^k \lambda_i f(\xvec^{(i)})
    \end{array} \right\}, \]
and $\fhat(\xvec) \defeq \inf(U_f(\xvec))$ for all $\xvec \in [0, 1]^m$.
Then $\fhat: [0, 1]^m \to \mathbb{R}_{\ge 0}$ is called
the \emph{convex hull} or \emph{convex envelope} of $f$, and is denoted by $\conv(f)$.
\end{definition}

$\conv(f)$ is, intuitively, the \emph{largest} convex function $\fhat$ such that $\fhat(x) \le f(x)$ for all $x$.
We show that replacing a pricing function $\pvec$
by its convex envelope $\pvectild$ does not reduce the revenue.
See \cref{fig:convex-hull} for an example of a pricing function and its convex envelope.

\begin{restatable}[Convexification]{theorem}{thmConvexEnvDominates}
\label{thm:convex-env-dominates}
Consider a buyer with linear valuation function $\v: [0, 1]^m \to \mathbb{R}_{\ge 0}$ and budget $b$.
Let $\p$ be an MLC pricing function. Then
\begin{tightenum}
\item $r(\conv(p) \mid \v, b) \ge r(\p \mid \v, b)$.
\item $\conv(p)$ is monotone, continuous, and convex.
\end{tightenum}
\end{restatable}

\begin{figure}[htb]
\centering
\begin{tikzpicture}[scale=0.6]
\datavisualization [
school book axes,
visualize as line={my data},
my grid style,
my data={no lines},
x axis={length={6cm},label=$x$,ticks={step=0.2,minor steps between steps=1}},
y axis={length={6cm},label={price},ticks={step=1,minor steps between steps=1}},
]
data {
x,y
0,0
1,5
}
info {
\path[draw={red!75!textHeavy},thick] (visualization cs: x=0, y=0)
    to[out=10, in=225] (visualization cs: x=0.4, y=1)
    to[out=45, in=225] (visualization cs: x=0.5, y=2.5)
    to[out=45, in=225] (visualization cs: x=0.8, y=3)
    to[out=45, in=260] (visualization cs: x=1, y=5);
\path[draw={blue!75!textHeavy},thick,dashed] (visualization cs: x=0, y=0)
    to[out=10, in=225] (visualization cs: x=0.4, y=1)
    -- (visualization cs: x=0.8, y=3)
    to[out=45, in=260] (visualization cs: x=1, y=5);
\path[fill={textColor}] (visualization cs: x=0.4, y=1) circle [radius=1.5pt]
    (visualization cs: x=0.8, y=3) circle [radius=1.5pt];
};
\end{tikzpicture}

\caption[A pricing function and its convex hull]{%
A pricing function $\textcolor{red!75!textHeavy}{\pvec}$ (solid red)
and its convex hull $\textcolor{blue!75!textHeavy}{\pvectild}$ (dashed blue).
Note that $\pvec(x) \neq \pvectild(x)$ for $x \in (0.4, 0.8)$,
and $\pvectild$ is linear in that interval.}
\label{fig:convex-hull}
\end{figure}
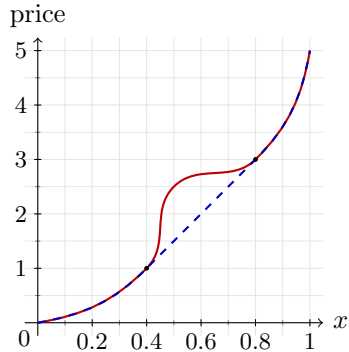

We formally prove \cref{thm:convex-env-dominates} in \cref{sec:price-char:convex-hull}.
To gain some insight into why \cref{thm:convex-env-dominates} holds, consider a buyer with infinite budget.
Now, the key observation is that \emph{the buyer's utility-maximizing bundle
cannot lie in a region where $\pvec$ strictly exceeds its convex envelope $\pvectild$,
that is, where $\pvec(\xvec) > \pvectild(\xvec)$}.
Intuitively, such regions correspond to the local non-convex portions of the pricing function, which can be ironed out.
(We compare this with popular ironing procedures in \cref{sec:ironing-review}.)
This observation follows from the following two ideas:
\begin{enumerate}
\item $\conv(\p - \v) = \conv(\p) - \v$, i.e., taking the convex envelope of a pricing function
    is equivalent to taking the concave envelope of the buyer's net utility.
    This follows easily from the definition of $\conv$ and the linearity of $\v$.
\item Any function $f$ is minimized at a point $\xvec$ if and only if
    $f(\xvec) = \conv(f)(\xvec)$ and $\xvec$ minimizes $\conv(f)$.
    Applying this to $\p - \v$ tells us that a bundle $\xvec$ is utility-maximizing for price $\p$
    iff it is utility maximizing for price $\pvectild$ and satisfies
    $(\p - \v)(\xvec) = \conv(\p - \v)(x) = (\pvectild - \v)(\xvec)$,
    which implies $\p(\xvec) = \pvectild(\xvec)$.
\end{enumerate}

Thus, $\Dem(\p \mid \v, \infty)$ is a subset of $\Dem(\pvectild \mid \v, \infty)$,
so $r(\p \mid \v, \infty) \le r(\pvectild \mid \v, \infty)$.
In \cref{sec:price-char:convex-hull}, we extend this argument to the setting where
the buyer has only a finite budget, through some additional non-trivial observations and insights.

\subsubsection{Step 2: Piecewise-Linearization}
\label{sec:overview:plin}

We now describe the second step: piecewise-linearization.
Unlike step 1 (convexification), this step requires the pricing function to be univariate,
i.e., we assume there is a single dataset.

For any function $f: [0, 1] \to \mathbb{R}$, let $f'(x)$ denote the left derivative of $f$ at any point $x \in (0, 1]$.
For any finite set $S \subseteq \mathbb{R}$, let $\disc(f, S)$ be another function $\fhat$
such that $\fhat(0) = f(0)$, and for all $x \in (0, 1]$, we have
$\fhat'(x) = \min(\{\alpha \in S: f'(x) \le \alpha\})$, i.e.,
$\fhat'(x)$ is obtained by rounding up $f'(x)$ to the smallest number in $S$.
Then $\fhat$ is called the \emph{discretization} or \emph{piecewise-linearization} of $f$.
See \cref{fig:plin} for an example of piecewise-linearizing a pricing function.

\begin{figure*}[htb]
\centering
\begin{tikzpicture}[scale=0.55,baseline={(current bounding box.center)}]
\datavisualization [
school book axes,
visualize as line={my data},
my grid style,
my data={no lines},
x axis={length={5cm},label=$x$,ticks={step=0.2}},
y axis={length={5cm},label=$p'(x)$,ticks={step=1}},
]
data {
x,y
0,0
1,6
}
info {
\path [draw={textColor!50!bgColor},dashed]
    (visualization cs: x=0.0, y=2)
    -- (visualization cs: x=1.0, y=2)
    (visualization cs: x=0.0, y=3)
    -- (visualization cs: x=1.0, y=3)
    (visualization cs: x=0.0, y=5)
    -- (visualization cs: x=1.0, y=5);
\path [draw={textBlue},very thick]
    (visualization cs: x=0.0, y=1)
    -- (visualization cs: x=0.2, y=1)
    (visualization cs: x=0.2, y=2)
    -- (visualization cs: x=0.4, y=3)
    -- (visualization cs: x=0.6, y=3)
    (visualization cs: x=0.6, y=4)
    -- (visualization cs: x=1.0, y=6);
\path [draw={textBlue},very thick,fill={bgColor}]
    (visualization cs: x=0.0, y=1) circle [radius=2pt]
    (visualization cs: x=0.2, y=2) circle [radius=2pt]
    (visualization cs: x=0.6, y=4) circle [radius=2pt];
};
\end{tikzpicture}

\qquad
\tikz[baseline={(current bounding box.center)}]
    \node at (0,0) {\huge$\longrightarrow$};
\quad\qquad
\begin{tikzpicture}[scale=0.55,baseline={(current bounding box.center)}]
\datavisualization [
school book axes,
visualize as line={my data},
my grid style,
my data={no lines},
x axis={length={5cm},label=$x$,ticks={step=0.2}},
y axis={length={5cm},label=$\phat'(x)$,ticks={step=1}},
]
data {
x,y
0,0
1,6
}
info {
\path [draw={textColor!50!bgColor},dashed]
    (visualization cs: x=0.0, y=2)
    -- (visualization cs: x=1.0, y=2)
    (visualization cs: x=0.0, y=3)
    -- (visualization cs: x=1.0, y=3)
    (visualization cs: x=0.0, y=5)
    -- (visualization cs: x=1.0, y=5);
\path [draw={textBlue},very thick]
    (visualization cs: x=0.0, y=2)
    -- (visualization cs: x=0.2, y=2)
    (visualization cs: x=0.2, y=3)
    -- (visualization cs: x=0.6, y=3)
    (visualization cs: x=0.6, y=5)
    -- (visualization cs: x=1.0, y=5);
\path [draw={textBlue},very thick,fill={bgColor}]
    (visualization cs: x=0.0, y=2) circle [radius=2pt]
    (visualization cs: x=0.2, y=3) circle [radius=2pt]
    (visualization cs: x=0.6, y=5) circle [radius=2pt];
};
\end{tikzpicture}

\caption[Piecewise-linearizing a pricing function]{%
The first plot shows the left derivative of a pricing function $p$.
The second plot shows the left derivative of the piecewise-linearized
$\phat \defeq \disc(p, \{2, 3, 5\})$.}
\label{fig:plin}
\end{figure*}
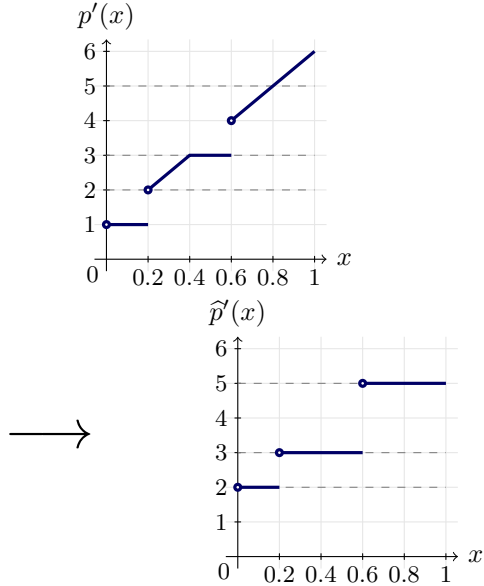

\begin{restatable}[Piecewise-linearization]{theorem}{thmDiscRevenue}
\label{thm:disc-revenue}
Consider a single dataset, and a buyer with linear valuation function $\v: [0, 1] \to \mathbb{R}_{\ge 0}$ and budget $b$.
Let $p: [0, 1] \to \mathbb{R}_{\ge 0}$ be a monotone, continuous, and convex pricing function.
Let $S \subset \mathbb{R}$ be a finite set such that $\v(1) \in S$. Let $\phat \defeq \disc(p, S)$. Then
\begin{tightenum}
\item $r(\phat \mid \v, b) \ge r(p \mid \v, b)$.
\item $\phat$ is monotone, continuous, convex, and piecewise-linear.
\end{tightenum}
\end{restatable}
\begin{proof}[Proof sketch]
We defer the full proof to \cref{sec:price-char:disc}.

Since $p$ is convex, the buyer purchases data records from the dataset only up to the rightmost point where
her marginal utility equals the marginal price. Beyond this point, convexity ensures that
the marginal price continues to rise, so any additional purchase can only reduce the buyer's net utility.

Since $\v(1) \in S$, we get that $p'(x) - \v'(x)$ and $\phat'(x) - \v'(x)$ do not have opposite signs for any $x \in (0, 1]$.
Thus, the buyer will buy the same amount for both pricing functions $p$ and $\phat$.
Since $\phat(x) \ge p(x)$ for all $x \in [0, 1]$, we get $r(\phat \mid \v, b) = r(\phat \mid \v, b)$.
\end{proof}

Now suppose we have multiple buyers, and each buyer $i \in [n]$ has
valuation function $\v_i: [0, 1] \to \mathbb{R}$ and budget $b_i$.
If we change the dataset's price to $\phat \defeq \disc(p, \{\v_i(1): i \in [n]\})$,
then \cref{thm:disc-revenue} tells us that the total revenue increases and $\phat$ is PLC.

\subsubsection{Composing the Transformations}

We would like to combine the two steps---convexification and piecewise-linearization---%
to prove the structural theorem (\cref{thm:struct}).
This is not straightforward when there are multiple datasets,
since the piecewise-linearization result, as stated above, only works for a single dataset.

To extend it to multiple datasets, we utilize the separability of the pricing function.
For any separable MLC pricing function $p: [0, 1]^m \to \mathbb{R}_{\ge 0}$,
let $\phat_j \defeq \disc(\conv(p_j), \{\v_{i,j}: i \in [n]\})$.
Then the separable pricing function $\phat$ having components $\phat_1, \ldots, \phat_m$
is what we need in \cref{thm:struct}.
This works because of several important properties of separable functions
that we prove in \cref{sec:price-char:sep}.
Finally, we prove the structural theorem (\cref{thm:struct}) in \cref{sec:price-char:together}.

\subsection{Finding the Optimal PLC Pricing Using an LP}
\label{sec:overview:lp}

Finally, we demonstrate that within the class of separable PLC functions,
an optimal pricing scheme can be determined via a linear program (LP). We now outline the construction.

For each dataset $j \in [m]$, let $\sigma_j$ be a permutation of the buyers $[n]$ such that
for each $k \in [n]$, $\sigma_i(k)$ is the buyer with the $k\Th$ smallest valuation for dataset $j$,
i.e., $\v_{\sigma_j(1),j} \le \ldots \le \v_{\sigma_j(n),j}$.
For each dataset $j \in [m]$, let $z_{\ell,j}$ represent the fraction of data records priced at slope $\v_{\ell,j}$.
Thus, the function $p_j(\cdot)$ increases with slope $\v_{\sigma_j(1),j}$ up to $z_{1,j}$,
then with slope $\v_{\sigma_j(2),j}$ from $z_{1,j}$ to $z_{1,j} + z_{2,j}$, and so on.

Intuitively, each buyer $i$ is willing to purchase records as long as it still has budget to spend
and the \emph{marginal utility} from an additional record exceeds its \emph{marginal price}.
Thus, the total expenditure of buyer $i$ is the minimum of his budget and the total price of
all data that yields a non-negative net utility, i.e.,
\begin{equation}
\label{eq:lp-rev}
r(p \mid \v_i, b_i) = \min\Bigl(b_i,\; \sum_{j=1}^m \sum_{t=1}^n \v_{t,j}\boolOne(\v_{t,j} \le \v_{i,j}) z_{t, j} \Bigr),
\end{equation}
where $\boolOne(P)$ is 1 if proposition $P$ is true and 0 otherwise.
One can prove this formally using \cref{thm:budget-exhaustion} (budget exhaustion) from \cref{sec:price-char:budget}
and \cref{thm:revenue-sep} from \cref{sec:price-char:sep}.

One can interpret a PLC-priced dataset as being broken into \emph{shards}, where each shard is linearly priced.
By this interpretation, each buyer is solving a fractional knapsack problem over the shards.
\Cref{sec:sharding-lp} describes this connection in more detail.

Now, observe that the total expenditure by buyers can be maximized through the following LP.
\begin{equation}
\label{eqn:lp}
\optprog{\max_{\substack{\zvec \in \mathbb{R}^{n \times m} \\ \rvec \in \mathbb{R}^n}}}{\sum_{i=1}^n r_i}{%
\\ \text{ where } & r_i \le b_i & \forall i \in [n] & \text{(budget constraint)}
\\ & r_i \le \sum_{j=1}^m \sum_{t=1}^n \v_{t,j}\boolOne(\v_{t,j} \le \v_{i,j}) z_{t, j}
    & \forall i \in [n] & \text{(by \eqref{eq:lp-rev})}
\\ & \sum_{t=1}^n z_{t,j} = 1 & \forall j \in [m] & \text{(shard sizes sum to 1)}
\\ & z_{t,j} \ge 0 & \forall t \in [n], \forall j \in [m] & \text{(non-negativity)}
}
\end{equation}

Thus, we can find the revenue-maximizing separable pricing function in polynomial time by solving LP \eqref{eqn:lp}.

\paragraph{Almost linear structure of the optimal PLC function.}
In addition to allowing revenue to be maximized in polynomial time,
the LP reveals a striking geometric property of the optimal PLC pricing functions%
---namely, that the optimal solution exhibits an \emph{almost linear} structure.

\begin{theorem}
\label{thm:lp-bfs}
Let $(\rvec, \zvec)$ be an optimal basic feasible solution to LP \eqref{eqn:lp}.
Then at most $m+n$ entries in $\zvec$ are positive.
\end{theorem}
\begin{proof}
LP \eqref{eqn:lp} has a bounded feasible region, since $0 \le r_i \le b_i$ for all $i \in [n]$,
and $0 \le z_{i,j} \le 1$ for all $i \in [n]$ and $j \in [m]$.
In any bounded linear program, there always exists an optimal solution that is \emph{basic feasible}.
In such a solution, the number of tight constraints (i.e., constraints satisfied with equality)
is at least the number of variables \citep{bertsimas1997introduction,bazaraa2010linear}.

Thus, in LP \eqref{eqn:lp}, at least $mn + n$ constraints are tight.
Suppose $k$ of the $2n$ constraints of the form $r_i \le \ldots$ are tight.
Suppose $d$ entries in $\zvec$ are positive.
Then the number of tight constraints is $k + m + (mn - d)$.
So, $k + m + mn - d \ge mn + n$, which implies $d \le (k + m + mn) - (mn + n) = m + (k-n) \le m + n$.
\end{proof}

Thus, in an optimal basic feasible solution, the total number of \emph{kinks}
(points of non-differentiability) in the PLC pricing functions is at most $n$.
Thus, $\max(0, m-n)$ of the datasets are linearly priced.
When $m \gg n$, nearly all datasets are priced \emph{linearly}.

We now present an example with $n-1$ kinks, showing that our upper bound on the number of kinks is nearly tight.

\begin{restatable}{lemma}{thmMaxKinks}
\label{thm:max-kinks}
Consider a data pricing instance with $n$ buyers and a single dataset.
Each buyer $i \in [n]$ values the dataset at $i$,
and her budget is $b_i \defeq \sum_{j=1}^i j/n = i(i+1)/2n$.
Let $z^*_{i,1} \defeq 1/n$ for all $i \in [n]$.
Then $(\bvec, \zvec^*)$ is the unique optimal solution to LP \eqref{eqn:lp}.
\end{restatable}
\begin{proof}
Proof deferred to \cref{sec:price-char:kinks}.
\end{proof}

\subsection{Linearity Gap}
\label{sec:overview:lingap}

Since the optimal pricing function is nearly linear, it is natural to ask
how much revenue is sacrificed when all pricing functions are required to be linear.
We formalize this as the \emph{linearity gap}:
the maximum ratio between the revenues of an optimal PLC and a linear pricing scheme, i.e.,
\begin{align}
\text{linearity gap} =
    \frac{\displaystyle \max_{\p,\,\xvec_i \in \Dem_i(\pvec)\,\forall i}\;\sum_i \pvec(\xvec_i)}{%
        \displaystyle \;\max_{\p:\text{linear},\,\xvec_i \in \Dem_i(\pvec)\,\forall i}\;\sum_i \pvec(\bm x_i)\;}
\end{align}
Note that if we add a constraint in LP \eqref{eqn:lp} insisting on the integrality of $\zvec$,
then the resulting mixed-integer LP gives us the revenue-maximizing linear pricing.
Hence, the linearity gap is precisely the integrality gap of LP \eqref{eqn:lp}.

We now give an example where the linearity gap is $2 - 1/n$.
The key idea is to have a single dataset, one \emph{rich} buyer, and $n-1$ \emph{poor} buyers.
The rich buyer has a large value and budget.
The poor buyers are identical and have small values and budgets.
With linear pricing, we can either set a low price and collect revenue from all the buyers,
or we can set a high price and collect a large revenue from the rich buyer and no revenue from the poor buyers.
With non-linear pricing, this dilemma disappears:
we can set a low price for a prefix of the data and a high price for the remainder such that every buyer exhausts her budget.

\begin{example}[linearity gap of $2-1/n$]
\label{ex:overview:lin-gap}
Consider an instance with $n \ge 2$ buyers and a single dataset.
Let $\eps > 0$ be an infinitesimally small number.
For all $i \in [n-1]$, we have $\v_{i,1} = \eps$ and $b_i = \eps(1-\eps)$.
We have $\v_{n,1} = (n-1)(1-\eps)$ and $b_n = n\eps(1-\eps)$.

Consider the pricing function $p: [0, 1] \to \mathbb{R}_{\ge 0}$ such that
$p'(x) = \v_{1,1}$ for $0 < x \le 1-\eps$ and $p'(x) = \v_{n,1}$ for $1-\eps < x \le 1$.
Then $p(1-\eps) = b_1$ and $p(1) = b_n$.
Then the revenue from $p$ is $\sum_{i=1}^n b_i = (2n-1)\eps(1-\eps)$.
Since buyers spend all their budget, $p$ is revenue-optimal.

Among linear prices, we only need to look at $\v_{1,1}$ and $\v_{n,1}$
(by Lemma~1 in \cite{chaudhury2026revenue},
or by \cref{thm:disc-revenue} in \cref{sec:price-char:disc}).
With a linear price of $\v_{1,1}$, the revenue is
$(n-1)\min(b_1, \v_{1,1}) + \min(b_n, \v_{1,1}) = (n-1)\eps(1-\eps) + \eps = \eps(n(1-\eps) + \eps)$.
With a linear price of $\v_{n,1}$, the revenue is $\min(b_n, \v_{n,1}) = n\eps(1-\eps)$.
The former is larger. Hence, the linearity gap is
\[ \frac{(2n-1)\eps(1-\eps)}{\eps(n(1-\eps)+\eps)} = \frac{2n-1}{n + \frac{\eps}{1-\eps}}
    \approx 2 - \frac{1}{n}. \]
\end{example}

\section{Empirics}
\label{sec:empirics}

In this section, we empirically study the properties of the
solution output by the LP-based algorithm.
In particular, we construct data pricing instances using real-world datasets,
and study the linearity gap and the number of kinks in the pricing function.

We construct data pricing instances using the California Housing dataset \cite{caliHousingDataset}. It contains 20640 rows and 9 numeric columns. Each row corresponds to a 1990 US Census block, and each column corresponds to a summary statistic on the block (e.g., median house age, average number of rooms per house, population of the block). One of the 9 numeric columns is \emph{median house value} (the \emph{target column}), and the remaining are \emph{input columns}.
We assume that buyers want to learn to predict the target column using the input columns.

\paragraph{Constructing the valuation matrix.}
Let $R$ be the set of census blocks (rows) and $C$ be the set of input columns.
Assume that the target column is known to all buyers.
We partition $C$ into $m = 4$ supposedly-independent categories:
(i) income: $C_1 =$ \{median income\},
(ii) structure: $C_2 =$ \{house age, average number of rooms, average number of bedrooms\},
(iii) demographics: $C_3 =$ \{population, average occupancy\},
(iv) location: $C_4 =$ \{latitude, longitude\}.
We randomly select 80\% of the rows into a \emph{training set} $R_{\train}$
and the remaining rows comprise the \emph{test set} $R_{\test}$.
For each $j \in [m]$, the seller offers a dataset $\Dcal_j$,
comprising of entries in $R_{\train} \times C_j$.

We sort the rows of the test set by latitude,
and then assign the first $|R_{\test}|/n$ rows to buyer 1,
the next $|R_{\test}|/n$ rows to buyer 2, and so on.
Let $R_{\test,i}$ be the rows assigned to each buyer $i \in [n]$.
Intuitively, $R_{\test,i}$ represents the geographical region where
buyer $i$ is interested in predicting the median house value.

Each buyer $i \in [n]$ estimates the value of dataset $\Dcal_j$
by how well it can predict housing prices on $R_{\test,i}$.
Specifically, we train a ridge regression model on $\Dcal_j$,
and $\v_{i,j}$ is the marginal gain in the precision (reciprocal of the mean squared error) due to $\Dcal_j$.
Since it is not known how to find the optimal prices for non-linear valuations,
we work with the \emph{linearized} instance, i.e., we assume that
the marginal gain in precision due to a bundle $\zvec$ is given by $\v_i(\zvec) = \sum_{j=1}^m \v_{i,j}z_j$.

\paragraph{Budgets.}
For each buyer $i \in [n]$, we set her budget $b_i$ to $B_i \cdot f_{\max} \cdot (\sum_{j=1}^m \v_{i,j})$.
Here $B_i \sim Beta(\alpha_B, 1)$, and $f_{\max}$ and $\alpha_B$ are positive constants.
(One can show that $\Beta(\alpha, 1)$ is the distribution of $X^{1/\alpha}$ for $X \sim \Unif(0, 1)$.)

\paragraph{Computation.}
We pick several values of $n \in [1, 100]$, and for each value,
we construct 16 instances by varying the other parameters
$f_{\max} \in [0.4, 1.3]$ and $\alpha_B \in [0.3, 1]$.
For each of these instances,
\begin{enumerate}
\item we find the revenue-maximizing prices by solving the LP using the dual-simplex method,
    and we observe the number of kinks.
\item for $n \le 20$, we compute the revenue-maximizing linear prices
    using the $O(n^m\cdot nm)$-time brute-force algorithm \citep{chaudhury2026revenue}.
    We then compute the linearity gap.
\end{enumerate}
These computations ran in 16 seconds on an Apple M4 processor (MacBook Air 2025).
The algorithms were implemented in a mix of pure Python and Numpy.
We solved linear programs using the HiGHS dual-simplex solver via SciPy.

\begin{figure}[htb]
\centering
\includegraphics[width=0.48\linewidth]{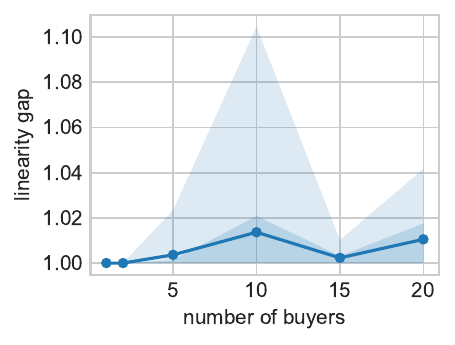}
\hfill
\includegraphics[width=0.48\linewidth]{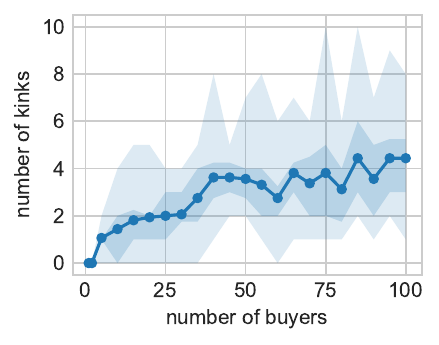}
\caption[Linearity gap and number of kinks]{%
Linearity gap and number of kinks for instances generated from the California Housing dataset. For a given value of $n$, there are 16 data pricing instances. The line represents the mean. The heavily shaded region lies between the first and third quartiles. The lightly-shaded region is between the min and max.}
\label{fig:caliH}
\end{figure}

\paragraph{Results.}
\cref{fig:caliH} summarizes our results.
It shows that the linearity gap is at most 1.1.
This is significantly less than $2-1/n$, which is the largest linearity gap we could demonstrate (\cref{ex:overview:lin-gap}). This indicates that the linearity gap is small in practice,
and motivates us to make the following conjecture.

\begin{conjecture}
For a monopolistic data pricing problem, the linearity gap is at most $2$,
i.e., linear pricing can recover at least half the revenue
of the optimal (potentially non-linear) pricing.
\end{conjecture}

In \cref{fig:caliH}, we also see that the number of kinks is at most 10.
On the other hand, in \cref{thm:max-kinks},
we give an example where the number of kinks is $n-1$.
This shows that in practice, the number of kinks can be much smaller than the worst case.

\phantomsection
\addcontentsline{toc}{section}{References}
\ifx\compatmode\undefined

\else
\fi
\input{dm-monopoly-final.bbl}

\clearpage

\appendix
\crefalias{section}{appendix}
\crefalias{subsection}{appendix}
\crefalias{subsubsection}{appendix}

\section{Details on Revenue Maximization under General Pricing}
\label{sec:rev_max}
\label{sec:price-char}

In this section, we provide extra insights and missing proofs from \cref{sec:overview}.

\paragraph{Notation.}
Recall \cref{defn:revenue}. Given a pricing function $\p$, a valuation function $\v$, and a budget $b$,
it formally defines the feasible set $F(\p \mid \v, b)$, the maximum net utility $u^*(\p \mid \v, b)$,
the demand set $\Dem(\p \mid \v, b)$, and the revenue $r(\p \mid \v, b)$.
In proofs, when $\v$ and $b$ are clear from context, we simply write
$F(\p)$, $u^*(\p)$, $\Dem(\p)$, and $r(\p)$.

We adopt the convention of denoting pricing functions by $\p$, $\pvechat$, etc.~and
denoting general functions by $f$, $\fhat$, etc.
Note that, unlike pricing functions, general functions
may not be monotone and may not satisfy $f(\vecZero) = 0$.

\paragraph{Organization.}
We partition our proofs into two categories:
(i) results specific to the sale of data, and
(ii) well-known or elementary results in real and convex analysis.
The former are in this appendix, and the latter are in \cref{sec:real-analysis}.
Whenever we use a result from \cref{sec:real-analysis}, we either restate it here
or describe it in one line, so that this appendix can be read on its own.
This appendix is organized as follows.
\begin{tightenum}
\item \cref{sec:price-char:mlc} explains why we require pricing functions to be MLC.
\item \cref{sec:price-char:budget} proves the budget exhaustion theorem (\cref{thm:budget-exhaustion}),
    which relates the revenue from a budgeted buyer to that from an unbudgeted buyer.
\item \cref{sec:price-char:convex-hull} proves the convexification theorem (\cref{thm:convex-env-dominates}).
\item \cref{sec:price-char:rate-eq} characterizes the demand set of
    a convex pricing function for a single dataset.
\item \cref{sec:price-char:disc} proves the piecewise-linearization theorem (\cref{thm:disc-revenue}).
\item \cref{sec:price-char:sep} studies separable functions, which lets us extend
    our single-dataset results to multiple datasets.
\item \cref{sec:price-char:together} combines the above to prove the structural theorem (\cref{thm:struct}).
\item \cref{sec:sharding-lp} interprets a PLC-priced dataset as being broken into \emph{shards}.
\item \cref{sec:price-char:kinks} proves \cref{thm:max-kinks}, which gives an instance
    where the optimal PLC pricing function has $n-1$ kinks.
\end{tightenum}

\subsection{MLC Pricing}
\label{sec:price-char:mlc}

We remark that the demand set may be empty, i.e.,
there may not exist a utility-maximizing bundle for the buyer,
as \cref{ex:no-util-max} demonstrates.
Then the revenue $r(\p \mid \v, b)$ is not well-defined.

\begin{example}[No utility-maximizing bundle]
\label{ex:no-util-max}
There is a single dataset. Let $\v(x) = 2x$ and $b = \infty$.
Let $p(x)$ be $x$ if $x < 1/2$ and $2x$ if $x \ge 1/2$.
Then all bundles are affordable and
\[ \v(x) - p(x) = \begin{cases}
x & \text{ if } 0 \le x < 1/2,
\\ 0 & \text{ if } 1/2 \le x \le 1.
\end{cases} \]

\begin{figure*}[htb]
\centering
\begin{subfigure}[b]{0.45\textwidth}
\centering
\begin{tikzpicture}
\datavisualization [
school book axes,
my grid style,
visualize as line/.list={seg1,seg2},
every visualizer/.style={style={draw={textColor},very thick}},
x axis={length={2.5cm},ticks={step=0.5},label=$x$},
y axis={length={2.5cm},label=$\p(x)$},
all axes={grid},
]

data [set=seg1] {
x,y
0,0
0.5,0.5
}
data [set=seg2] {
x,y
0.5,1
1,2
}
info {
\path [draw={textColor},thick,fill={bgColor}] (visualization cs: x=0.5, y=0.5) circle [radius=2pt];
};
\end{tikzpicture}

\caption{Pricing function}
\label{fig:no-util-max:price}
\end{subfigure}
\begin{subfigure}[b]{0.45\textwidth}
\centering
\begin{tikzpicture}
\datavisualization [
school book axes={unit=0.5},
my grid style,
visualize as line/.list={seg1,seg2},
every visualizer/.style={style={draw={textColor}, very thick}},
all axes={grid},
x axis={length={2.5cm},label=$x$},
y axis={length={2.5cm},label=$v(x)-\p(x)$,max value=1,
},
]
data [set=seg1] {
x,y
0,0
0.5,0.5
}
data [set=seg2] {
x,y
0.5,0
1,0
}
info {
\path [draw={textRed},thick,fill={bgColor}] (visualization cs: x=0.5, y=0.5) circle [radius=2pt];
\node[left] at (visualization cs: x=-0.3, y=0.5) {\textcolor{textRed}{($u^*$)}};
\node[below] at (visualization cs: x=-0.5, y=0.40) {\textcolor{textRed}{\small unattainable}};
\path [draw={textRed},dashed]
    (visualization cs: x=0.0, y=0.5)
    -- (visualization cs: x=0.5, y=0.5);
};
\end{tikzpicture}

\caption{Net utility}
\label{fig:no-util-max:net-util}
\end{subfigure}
\caption{Pricing function and net utility for \cref{ex:no-util-max}.}
\label{fig:no-util-max}
\end{figure*}
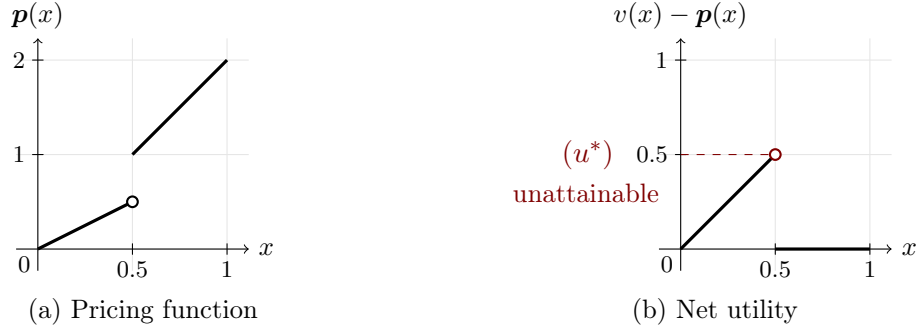
The plots of $p$ and $\v-p$ are illustrated in \cref{fig:no-util-max}.
Notice that, $\lim_{x \to 0.5^-} (\v-p)(x) = 0.5$, hence, the optimal utility of the buyer is $u^*(p) = 0.5$.
However, since $(\v-p)(x) < 0.5$ when $x \in [0, 0.5)$ and $(\v-p)(0.5) = 0$,
we get that $u^*$ is actually unattainable and the demand set $\Dem(p) = \emptyset$.
\end{example}

The discontinuity of $p$ in \cref{ex:no-util-max} seems to be the reason for
non-existence of a utility-maximizing affordable bundle.
Indeed, we show (in \cref{thm:mlc-demand-nonempty}) that
whenever a pricing function $\p$ is MLC (monotone lower-continuous),
a utility-maximizing affordable bundle always exists.
The formal definition of MLC and the formal proof are deferred to
\cref{sec:real-analysis:mlc} (the proof is on page~\pageref{prf:mlc-demand-nonempty}).

\begin{restatable}{lemma}{lemMLCDemandNonempty}
\label{thm:mlc-demand-nonempty}
Let $\v, \p: [0, 1]^m \to \mathbb{R}$ be functions where $\v$ is monotone and continuous,
$\p$ is MLC, and $\p(\vecZero) = 0$.
Let $b \in \mathbb{R}_{\ge 0}$. Let $F(\p) \defeq \{\xvec \in [0, 1]^m: \p(\xvec) \le b\}$.
Then $\exists \xvec^* \in F(\p)$ such that
\[ \v(\xvec^*) - \p(\xvec^*) = \sup_{\xvec \in F(\p)} (\v(\xvec) - \p(\xvec)). \]
\end{restatable}

Intuitively, assuming lower-continuity is without loss of generality.
This is because at any point $\xvec \in [0, 1]^m$ of discontinuity,
if the buyer selects the bundle $(1-\delta)\xvec$ instead of $\xvec$
for an infinitesimally small $\delta > 0$,
then her value for the bundle reduces by an infinitesimally small amount
(assuming $\v$ is continuous), but the price reduces significantly.
Hence, she would never purchase at $\xvec$, and so, we can replace the price at $\xvec$
by its \emph{lower limit} at $\xvec$ (formally defined in \cref{defn:continuity}), with no loss to the revenue.
Hence, we can assume \wLoG{} that $\p$ is lower-continuous.

\subsection{Budget Exhaustion}
\label{sec:price-char:budget}

Intuitively, a larger budget leads to more spending, so $r(\p \mid \v, b) \le r(\p \mid \v, \infty)$.
We also have the trivial bound $r(\p \mid \v, b) \le b$, since we cannot extract more revenue than the budget.
Combining these gives us $r(\p \mid \v, b) \le \min(b, r(\p \mid \v, \infty))$
(we prove this formally in \cref{thm:rev-budget-ub}).

When the pricing function is convex and continuous, and the valuation function is linear,
then we show that the inequality actually holds with equality.
We call this result the \emph{budget exhaustion} theorem (\cref{thm:budget-exhaustion}).
It is very helpful, since it allows us to prove results about the revenue
in the simpler unbudgeted setting and easily carry them over to the budgeted setting.

Before proving this theorem, we show that convexity is necessary (for linear valuations),
i.e., we give an example with non-convex pricing where the inequality may not be tight.

\begin{example}
\label{ex:non-convex-price}
Consider a single dataset, and a single buyer with valuation function $\v(x) = 2x$
and budget $b$. The pricing function is
\[ p(x) = \begin{cases}
x& \text{ if } 0 \le x \le 0.4,
\\ 3.5\cdot x - 1 & \text{ if } 0.4 \le x \le 0.6
\\x + 0.5 & \text{ if } 0.6 \le x \le 1
\end{cases}, \]
i.e., the per-unit cost (given by the derivative of $p$)
is 3.5 when $x \in (0.4, 0.6]$ and 1 otherwise.

\begin{figure*}[htb]
\centering
\begin{subfigure}[b]{0.31\textwidth}
\centering
\begin{tikzpicture}[scale=0.85]
\datavisualization [
school book axes,
visualize as line,
my grid style,
every visualizer/.style={style={draw={textColor}, very thick}},
all axes={ticks={step=0.5,minor steps between steps=4}},
x axis={length={4cm},label=$x$},
y axis={length={4cm},label=$\p(x)$},
]
data {
x,y
0.0,0.0
0.4,0.4
0.6,1.1
1.0,1.5
};
\end{tikzpicture}

\caption{Pricing function}
\label{fig:non-convex-price:price}
\end{subfigure}
\hfill
\begin{subfigure}[b]{0.35\textwidth}
\centering
\begin{tikzpicture}[scale=0.85]
\datavisualization [
school book axes,
visualize as line/.list={seg1,seg2,seg3},
seg1={style={draw={textColor},very thick}},
seg2={style={draw={textBlue},thick}},
seg3={style={draw={textBlue},dashed, thick}},
my grid style,
every visualizer/.style={style={draw={textColor},very thick}},
all axes={ticks={step=0.5,minor steps between steps=4}},
x axis={length={4cm},label=$x$},
y axis={length={4cm},label=$v(x)-\p(x)$,max value=1},
]
data [set=seg1] {
x,y
0.0,0.0
0.4,0.4
0.6,0.1
1.0,0.5
}
data [set=seg2] {
x,y
0.0,0.0
0.8,0.0
}
data [set=seg3] {
x,y
0.8,0.0
0.8,0.3
}
info {
\path [fill={textRed}] (visualization cs: x=0.4, y=0.4) circle [radius=3pt];
\path [fill={textGreen}] (visualization cs: x=1.0, y=0.5) circle [radius=3pt];
\node[above] at (visualization cs: x=0.4, y=0.4) {\textcolor{textRed}{\scriptsize local max}};
\node[above] at (visualization cs: x=1, y=0.5) {\textcolor{textGreen}{\scriptsize global max}};
\node[above] at (visualization cs: x=0.45, y=0) {\textcolor{textBlue}{\scriptsize affordable set}};
};
\end{tikzpicture}

\caption{Net utility $\v - p$}
\label{fig:non-convex-price:net-util}
\end{subfigure}
\hfill
\begin{subfigure}[b]{0.31\textwidth}
\centering
\begin{tikzpicture}[scale=0.85]
\datavisualization [
school book axes,
visualize as line/.list={seg1,seg2,seg3,seg4,seg5},
seg1={style={draw={textColor},very thick}},
seg2={style={draw={textColor},very thick}},
seg3={style={draw={textColor},very thick}},
seg4={style={draw={textBlue},dashed,thick}},
my grid style,
x axis={length={4cm},label=$x$,ticks={step=0.2}},
y axis={length={4cm},label=$\p'(x)$,ticks={step=1,minor steps between steps=1}},
]
data [set=seg1] {
x,y
0.0,1.0
0.4,1.0
}
data [set=seg2] {
x,y
0.4,3.5
0.6,3.5
}
data [set=seg3] {
x,y
0.6,1.0
1.0,1.0
}
data [set=seg4] {
x,y
0.0,2.0
1.0,2.0
}
data [set=seg5] {
x,y
0.0,4.0
}
info {
\path [draw={textColor},thick,fill={bgColor}] (visualization cs: x=0.4, y=3.5) circle [radius=2pt];
\path [draw={textColor},thick,fill={bgColor}] (visualization cs: x=0.6, y=1.0) circle [radius=2pt];
\path (visualization cs: x=1.0, y=2.0) node [anchor=west] {$\textcolor{textBlue}{v'(x)}$};
};
\end{tikzpicture}

\caption{Left derivative of price}
\label{fig:non-convex-price:price-d}
\end{subfigure}
\caption{Plots for \cref{ex:non-convex-price}.}
\label{fig:non-convex-price}
\end{figure*}

We can see that $x = 0.4$ is a local maximum for $\v - p$ and $x = 1$ is the global maximum.
At $x = 0.4$, buying a small amount of data decreases the buyer's net utility
(since the per-unit price is $3.5$, and the per-unit value is $2$).
However, buying the entire dataset is a good strategy for the buyer,
since the per-unit price is $1$ for $x > 0.6$,
so the loss incurred by purchasing the segment $(0.4, 0.6]$
is offset by the gain incurred by purchasing $(0.6, 1]$.
Hence, the buyer needs to \emph{think globally} to optimize her net utility.

For $b \ge 1.5$, the affordable set is $[0, 1]$, the demand set is $\{1\}$, and the revenue is $1.5$.
For $b = 1.3$, the affordable set is $[0, 0.8]$, the demand set is $\{0.4\}$, and the revenue is $0.4$.
Hence, the buyer is limited by her budget, but still does not exhaust it.
\end{example}

We now return to proving the budget exhaustion theorem.
We begin by showing (in \cref{thm:affordable-optimum}) that if the maximizer of $\v-\p$ is affordable,
i.e., $r(\p \mid \v, \infty) \le b$, then the maximum net utilities and the demand sets
are the same regardless of whether the budget is $b$ or $\infty$,
which proves that $r(\p \mid \v, b) = r(\p \mid \v, \infty)$.

\begin{lemma}[affordable optimum]
\label{thm:affordable-optimum}
Let $\v: [0, 1]^m \to \mathbb{R}_{\ge 0}$ be a monotone continuous valuation function,
let $\p: [0, 1]^m \to \mathbb{R}_{\ge 0}$ be an MLC pricing function,
and let $b \in \mathbb{R}_{\ge 0}$.
If $r(\p \mid \v, \infty) \le b$, then
$\Dem(\p \mid \v, b) = \Dem(\p \mid \v, \infty)$
and $r(\p \mid \v, b) = r(\p \mid \v, \infty)$.
\end{lemma}
\begin{proof}
$F(\p \mid b) \subseteq F(\p \mid \infty) = [0, 1]^m$, so $u^*(\p \mid \v, b) \le u^*(\p \mid \v, \infty)$.
By \cref{thm:mlc-demand-nonempty}, $\Dem(\p \mid \v, \infty)$ is non-empty.
Let $\xvec^* \in \Dem(\p \mid \v, \infty)$.
Since $\p(\xvec^*) \le r(\p \mid \v, \infty) \le b$, we get $\xvec^* \in F(\p \mid b)$,
so $u^*(\p \mid \v, \infty) = \v(\xvec^*) - \p(\xvec^*) \le u^*(\p \mid \v, b)$.
So, $u^*(\p \mid \v, b) = u^*(\p \mid \v, \infty)$,
and thus, $\Dem(\p \mid \v, b) \subseteq \Dem(\p \mid \v, \infty)$.
Moreover, $\xvec^* \in \Dem(\p \mid \v, b)$.
Since we picked $\xvec^*$ arbitrarily from $\Dem(\p \mid \v, \infty)$,
we get $\Dem(\p \mid \v, \infty) \subseteq \Dem(\p \mid \v, b)$.
Hence, the demand sets are equal for budgets $b$ and $\infty$,
and consequently, $r(\p \mid \v, b) = r(\p \mid \v, \infty)$.
\end{proof}

Next, we formally prove that $\min(b, r(\p \mid \v, \infty))$ is an upper bound on $r(\p \mid \v, b)$.

\begin{lemma}
\label{thm:rev-budget-ub}
Let $\v: [0, 1]^m \to \mathbb{R}_{\ge 0}$ be a monotone continuous valuation function,
let $\p: [0, 1]^m \to \mathbb{R}_{\ge 0}$ be an MLC pricing function,
and let $b \in \mathbb{R}_{\ge 0}$.
Then $r(\p \mid \v, b) \le \min(b, r(\p \mid \v, \infty))$.
\end{lemma}
\begin{proof}
Every bundle in $\Dem(\p \mid \v, b)$ lies in $F(\p \mid b)$, and so costs at most $b$.
Hence, $r(\p \mid \v, b) \le b$.
If $r(\p \mid \v, \infty) \le b$, then by \cref{thm:affordable-optimum}, we get
$r(\p \mid \v, b) = r(\p \mid \v, \infty) = \min(b, r(\p \mid \v, \infty))$.
Otherwise, $r(\p \mid \v, \infty) > b$, so
$\min(b, r(\p \mid \v, \infty)) = b \ge r(\p \mid \v, b)$.
\end{proof}

Finally, we prove the budget exhaustion theorem, i.e.,
$r(\p \mid \v, b) = \min(b, r(\p \mid \v, \infty))$ when $\v - \p$ is concave.
Intuitively, this holds because if the buyer demands a bundle $\xvec^*$ at a budget of $\infty$
that costs more than $b$, and bundle $\xvec^{(0)}$ at a budget of $b$,
then a convex combination of these bundles is also in the demand set and costs exactly $b$.

\begin{theorem}[budget exhaustion]
\label{thm:budget-exhaustion}
Let $\v: [0, 1]^m \to \mathbb{R}_{\ge 0}$ be a monotone continuous valuation function,
let $\p: [0, 1]^m \to \mathbb{R}_{\ge 0}$ be a continuous pricing function,
and let $b \in \mathbb{R}_{\ge 0}$.
If $\v-\p$ is concave, then $r(\p \mid \v, b) = \min(b, r(\p \mid \v, \infty))$.
\end{theorem}
\begin{proof}
Since $\p$ is a continuous pricing function, it is MLC.
If $r(\p \mid \v, \infty) \le b$, then by \cref{thm:affordable-optimum}, we get
$r(\p \mid \v, b) = r(\p \mid \v, \infty) = \min(b, r(\p \mid \v, \infty))$.

Now suppose $r(\p \mid \v, \infty) > b$.
We will prove that $r(\p \mid \v, b) = b$, i.e., the buyer exhausts her budget.
Since $r(\p \mid \v, \infty) = \sup_{\xvec \in \Dem(\p \mid \v, \infty)} \p(\xvec) > b$,
we can pick $\xvec^* \in \Dem(\p \mid \v, \infty)$ such that $\p(\xvec^*) > b$.
By \cref{thm:mlc-demand-nonempty}, we can pick $\xvec^{(0)} \in \Dem(\p \mid \v, b)$.
Then $\p(\xvec^{(0)}) \le b$.
For any $t \in [0, 1]$, define $h(t) \defeq \p\left((1-t)\xvec^{(0)} + t\xvec^*\right)$.
Since $\p$ is continuous, $h$ is also continuous. $h(0) \le b$ and $h(1) > b$.
By the intermediate value theorem (\cref{thm:intermed-value} in \cref{sec:real-analysis:mlc}),
there exists an $\alpha \in [0, 1]$ such that $h(\alpha) = b$.
Let $\xvechat \defeq (1-\alpha)\xvec^{(0)} + \alpha \xvec^*$. Then $\p(\xvechat) = b$.

Since $F(\p \mid b) \subseteq F(\p \mid \infty) = [0, 1]^m$, we have $u^*(\p \mid \v, b) \le u^*(\p \mid \v, \infty)$.
Hence, $(\v-\p)(\xvec^{(0)}) = u^*(\p \mid \v, b)$ and
$(\v-\p)(\xvec^*) = u^*(\p \mid \v, \infty) \ge u^*(\p \mid \v, b)$.
By concavity of $\v-\p$, we get
\[ (\v-\p)(\xvechat) \ge (1-\alpha)\cdot(\v-\p)(\xvec^{(0)}) + \alpha\cdot(\v-\p)(\xvec^*) \ge u^*(\p \mid \v, b). \]
Hence, $\xvechat \in \Dem(\p \mid \v, b)$ and $\p(\xvechat) = b$,
which implies $r(\p \mid \v, b) = b = \min(b, r(\p \mid \v, \infty))$.
\end{proof}

\subsection{Convexifying the Pricing Function}
\label{sec:price-char:convex-hull}

In this subsection, we show that any pricing function can be \emph{convexified},
i.e., we prove that replacing a pricing function with its \emph{convex hull} does not decrease revenue.

The convex hull of a function $f$ is, intuitively, the \emph{largest} convex function $\fhat$
such that $\fhat(x) \le f(x)$ for all $x$ (see \cref{defn:convex-env}
in \cref{sec:overview:convexify} for a formal definition).
\Cref{fig:convex-hull} in \cref{sec:overview:convexify} shows an example pricing function and its convex hull,
and gives intuition behind why convexification doesn't decrease the revenue for this pricing function.
We now formalize the intuition developed there.
Specifically, for any pricing function $\p$ and its convex hull $\pvechat$,
we show (in \cref{thm:convex-env-witness,thm:convex-env-v-minus-p-dom})
that any point $\xvec \in [0, 1]^m$ for which $\p(\xvec) \neq \pvechat(\xvec)$
can be written as a convex combination of multiple points where the values of $\p$ and $\pvechat$ coincide.
Then the buyer will always prefer one of these points instead of $\xvec$.

We begin by showing that if $\pvechat$ is the convex hull of an MLC pricing function $\p$,
then $\pvechat$ is monotone, convex, and continuous.
Then, we prove that the revenue of $\pvechat$ is at least that of $\p$.

Below we use notation $f$ and $\fhat$ when referring to general functions,
and $\p$ and $\pvechat$ when referring specifically to pricing functions.
Define $\Delta_n \defeq \{\xvec \in [0, 1]^n: \sum_{i=1}^n x_i = 1\}$ as the $n$-simplex.

\begin{restatable}{lemma}{thmConvexEnvProps}
\label{thm:convex-env-props}
Let $f: [0, 1]^m \to \mathbb{R}_{\ge 0}$ be a function and $\fhat$ be its convex hull. Then
\begin{tightenum}
\item $\fhat(\xvec) \le f(\xvec)$ for all $\xvec \in [0, 1]^m$.
\item $\fhat$ is convex.
\item If $f$ is monotone, then $\fhat$ is also monotone.
\item If $f$ is MLC, then $\fhat$ is continuous.
\end{tightenum}
\end{restatable}
\begin{proof}
The proof is deferred (page~\pageref{prf:convex-env-props})
to \cref{sec:real-analysis:convex-functions}.
\end{proof}

We now show a crucial property of the convex hull in \cref{thm:convex-env-witness}.
For any function $f$ and its convex hull $\fhat \defeq \conv(f)$,
let $Y \defeq \{(\xvec, \fhat(\xvec)): \xvec \in [0, 1]^m \text{ and } f(\xvec) = \fhat(\xvec)\}$.
Then any point $(\xvec, \fhat(\xvec))$ can be represented as a convex combination of points in $Y$.

\begin{lemma}
\label{thm:convex-env-witness}
Let $f: [0, 1]^m \to \mathbb{R}$ be an MLC function, $\fhat \defeq \conv(f)$, and $\xvechat \in [0, 1]^m$.
Then $\exists k \in \mathbb{N}$, $\exists \lambda \in \Delta_k$,
and $\exists \xvec^{(1)}, \ldots, \xvec^{(k)} \in [0, 1]^m$ such that
$f(\xvec^{(i)}) = \fhat(\xvec^{(i)})$ for all $i \in [k]$,
$\xvechat = \sum_{i=1}^k \lambda_i\xvec^{(i)}$,
and $\fhat(\xvechat) = \sum_{i=1}^k \lambda_if(\xvec^{(i)})$.
\end{lemma}
\begin{proof}
Since $f$ is MLC, one can show that $U_f(\xvec)$ has a minimum element
(see \cref{thm:Ufx-has-min} in \cref{sec:real-analysis:convex-functions}),
and so, $\fhat(\xvechat) \in U_f(\xvechat)$.
Then $\exists k \in \mathbb{N}$, $\exists \lambda \in \Delta_k$,
and $\exists \xvec^{(1)}, \ldots, \xvec^{(k)} \in [0, 1]^m$ such that
$\xvec = \sum_{i=1}^k \lambda_i \xvec^{(i)}$ and $\fhat(\xvechat) \ge \sum_{i=1}^k \lambda_i f(\xvec^{(i)})$.

By \cref{thm:convex-env-props}, $\fhat$ is convex.
Hence, $\fhat(\xvechat) \le \sum_{i=1}^k \lambda_i\fhat(\xvec^{(i)})$. Therefore,
\[ \sum_{i=1}^k \lambda_i\fhat(\xvec^{(i)}) \ge \fhat(\xvechat) \ge \sum_{i=1}^k \lambda_i f(\xvec^{(i)})
\implies \sum_{i=1}^k \lambda_i(f(\xvec^{(i)}) - \fhat(\xvec^{(i)})) \le 0. \]
By \cref{thm:convex-env-props}, $f(\xvec) \ge \fhat(\xvec)$ for all $\xvec \in [0, 1]^m$.
Hence, $f(\xvec^{(i)}) = \fhat(\xvec^{(i)})$ for all $i \in [k]$.
\end{proof}

We now use \cref{thm:convex-env-witness} to show that for pricing functions $\p$ and $\pvechat \defeq \conv(\p)$,
a buyer with infinite budget has a point $\xvec^*$ in her demand set where $\p(\xvec^*) = \pvechat(\xvec^*)$.

\begin{lemma}
\label{thm:convex-env-v-minus-p-dom}
Consider a buyer with linear valuation function $\v: [0, 1]^m \to \mathbb{R}_{\ge 0}$.
Let $\p$ be an MLC pricing function, $\pvechat \defeq \conv(\p)$, and $\xvechat \in [0, 1]^m$.
Then $\exists \xvec^* \in [0, 1]^m$ such that $\p(\xvec^*) = \pvechat(\xvec^*)$
and $\v(\xvec^*) - \pvechat(\xvec^*) \ge \v(\xvechat) - \pvechat(\xvechat)$.
\end{lemma}
\begin{proof}
By \cref{thm:convex-env-witness}, $\exists k \in \mathbb{N}$, $\exists \lambda \in \Delta_k$,
and $\exists \xvec^{(1)}, \ldots, \xvec^{(k)} \in [0, 1]^m$ such that
$\p(\xvec^{(i)}) = \pvechat(\xvec^{(i)})$ for all $i \in [k]$,
$\xvechat = \sum_{i=1}^k \lambda_i\xvec^{(i)}$,
and $\pvechat(\xvechat) = \sum_{i=1}^k \lambda_i\pvechat(\xvec^{(i)})$.
By linearity of $\v$, we get
\[ \v(\xvechat) - \pvechat(\xvechat) = \sum_{i=1}^k \lambda_i(\v(\xvec^{(i)}) - \pvechat(\xvec^{(i)})). \]
Then $x^*=x_t$ where $t \defeq \argmax_{i=1}^k (\v(\xvec^{(i)}) - \pvechat(\xvec^{(i)}))$ satisfies $\v(\xvec^*) - \pvechat(\xvec^*) \ge \v(\xvechat) - \pvechat(\xvechat)$.
\end{proof}

Finally, we show that the revenue of any pricing function $\p$
is at most the revenue of $\conv(\p)$.

\thmConvexEnvDominates*
\begin{proof}
Let $\pvechat \defeq \conv(\p)$.
Since $\p$ is MLC, by \cref{thm:convex-env-props},
$\pvechat$ is monotone, continuous, and convex,
and $\pvechat(\xvec) \le \p(\xvec)$ for all $\xvec \in [0, 1]^m$.

Since $\v$ is linear, $\v - \pvechat$ is concave and continuous.
If $r(\pvechat \mid \v, \infty) \ge b$, then by \cref{thm:budget-exhaustion},
we get $r(\pvechat \mid \v, b) = b \ge r(\p \mid \v, b)$, so we are done.
So now assume $r(\pvechat \mid \v, \infty) < b$.

Recall the definitions of $F$, $u^*$, $\Dem$, and $r$ from \cref{defn:revenue}.
Since $\pvechat$ is MLC, \cref{thm:mlc-demand-nonempty} tells us that
$\Dem(\pvechat \mid \v, \infty)$ is non-empty.
Let $\xvec^* \in \Dem(\pvechat \mid \v, \infty)$.
Then $\v(\xvec^*) - \pvechat(\xvec^*) \ge \v(\xvec) - \pvechat(\xvec)$ for all $\xvec \in [0, 1]^m$.
By \cref{thm:convex-env-v-minus-p-dom}, we can assume \wLoG{} that $\p(\xvec^*) = \pvechat(\xvec^*)$,
since the resulting bundle also maximizes $\v - \pvechat$ over $[0, 1]^m$
and so also lies in $\Dem(\pvechat \mid \v, \infty)$.
Since $r(\pvechat \mid \v, \infty) < b$, \cref{thm:affordable-optimum} gives us
$\Dem(\pvechat \mid \v, \infty) = \Dem(\pvechat \mid \v, b)$, so $\xvec^* \in \Dem(\pvechat \mid \v, b)$.

We next show that $\Dem(\p \mid \v, b) \subseteq \Dem(\pvechat \mid \v, b)$,
which implies $r(\pvechat \mid \v, b) \ge r(\p \mid \v, b)$.

Let $\xvechat \in \Dem(\p \mid \v, b)$.
Suppose $\pvechat(\xvechat) < \p(\xvechat)$. Then
\[ \v(\xvechat) - \p(\xvechat) < \v(\xvechat) - \pvechat(\xvechat)
    \le \v(\xvec^*) - \pvechat(\xvec^*) = \v(\xvec^*) - \p(\xvec^*). \]
This is a contradiction, since $\xvechat$ maximizes $\v-\p$ among all affordable bundles.
Hence, $\pvechat(\xvechat) = \p(\xvechat)$.

Since $\xvechat \in \Dem(\p \mid \v, b)$ and $\xvec^* \in \Dem(\pvechat \mid \v, b)$, we get
\[ \v(\xvechat) - \p(\xvechat) = \v(\xvechat) - \pvechat(\xvechat)
\le \v(\xvec^*) - \pvechat(\xvec^*) = \v(\xvec^*) - \p(\xvec^*) \le \v(\xvechat) - \p(\xvechat). \]
Hence, $u^*(\p \mid \v, b) = u^*(\pvechat \mid \v, b)$,
so $\xvechat \in \Dem(\pvechat \mid \v, b)$.
Since we picked $\xvechat$ arbitrarily from $\Dem(\p \mid \v, b)$,
we get $\Dem(\p \mid \v, b) \subseteq \Dem(\pvechat \mid \v, b)$. Hence,
\[ r(\pvechat \mid \v, b) = \sup_{\xvec \in \Dem(\pvechat \mid \v, b)} \p(\xvec)
    \ge \sup_{\xvec \in \Dem(\p \mid \v, b)} p(\xvec) = r(\p \mid \v, b).
    \qedhere \]
\end{proof}

\subsection{Rate Equalization}
\label{sec:price-char:rate-eq}

In this subsection and the next, we assume that there is a single dataset, i.e., the pricing function is univariate.
(We will remove this assumption in \cref{sec:price-char:sep} by exploiting the pricing function's separability.)

Consider a single dataset and a single buyer having a linear valuation function $\v$ and infinite budget.
If the pricing function $p: [0, 1] \to \mathbb{R}_{\ge 0}$ is convex,
then $\v(\cdot) - p(\cdot)$ is concave, so local maxima are global maxima too.
Hence, the demand set can be determined easily based on where the derivatives of $p$
and $\v$ intersect, as shown in \cref{ex:convex-price-derivative}.

\begin{example}
\label{ex:convex-price-derivative}
Consider a single dataset, and a single buyer having valuation function $\v(x) = x$.
Let $p(x) = x^2$ be the pricing function. Then the derivative of $p$ is $2x$.
The derivative gives us the per-unit price of buying an additional infinitesimal amount of data.
The per-unit valuation of the dataset is 1 (the derivative of $\v$).
Hence, the buyer would purchase as long as either the per-unit price and value equalize (at $x = 1/2$)
or the buyer exhausts her budget (which would happen if the budget is less than $1/4$).
\end{example}

We now formalize this observation.
We will show that the revenue-maximizing point in the buyer's demand set
is the rightmost point for which the \emph{left derivative} of $p$ is at most $\v(1)$.
To do this, we first precisely define this rightmost point.

For any MLC convex function $f: [0, 1] \to \mathbb{R}$ and $\beta \in \mathbb{R}$,
let $f'(x)$ be the left derivative at $x \in (0, 1]$, and let $f'(0) \defeq -\infty$.
Let $\ell(f \mid \beta) \defeq \max(\{x \in [0, 1]: f'(x) \le \beta\})$
be the rightmost point for which the left derivative of $f$ is at most $\beta$.
One can show that $\ell(f \mid \beta)$ is well-defined, i.e.,
$\{x \in [0, 1]: f'(x) \le \beta\}$ has a maximum element.
This is because the left derivative of an MLC convex function is also MLC
(see \cref{thm:convex-ld-mlc} in \cref{sec:real-analysis:ld}),
and $\{x \in [0, 1]: g(x) \le \beta\}$ is a closed set for any MLC function $g$
(see \cref{thm:mlc-feas-compact} in \cref{sec:real-analysis:mlc}).

For any MLC function $f: [0, 1] \to \mathbb{R}$, one can show that
$\ell(f \mid 0)$ is the minimizer of $f$ (\cref{thm:convex-ld-min} in \cref{sec:real-analysis:ld}).
Using this, we show that for any MLC convex pricing function $p$ for a single dataset
and any buyer with valuation function $\v$ and infinite budget,
$\ell(p \mid \v(1))$ is the revenue-maximizing point in the buyer's demand set.

\begin{lemma}
\label{thm:convex-price-rev}
Let $p: [0, 1] \to \mathbb{R}_{\ge 0}$ be an MLC convex pricing function.
Let $\v: [0, 1] \to \mathbb{R}_{\ge 0}$ be a linear function.
Let $x^* \defeq \ell(p \mid \v(1))$.
Then $x^* \in \Dem(p \mid \v, \infty)$ and $r(p \mid \v, \infty) = p(x^*)$.
\end{lemma}
\begin{proof}
Since the sum of convex functions is convex, the sum of lower-continuous functions is lower-continuous,
and the left derivative distributes over addition,
we get that $p-\v$ is convex and lower-continuous, and $(p-\v)'(x) = p'(x) - \v(1)$ for all $x \in [0, 1]$.
Thus, $x^* = \ell(p \mid \v(1)) = \ell(p - \v, 0)$.
By \cref{thm:convex-ld-min} in \cref{sec:real-analysis:ld}, we get that
$x^*$ minimizes $p - \v$, and no point in $(x^*, 1]$ minimizes $p - \v$.
Thus, $x^* \in \Dem(p \mid \v, \infty)$, and for any point $x \in \Dem(p \mid \v, \infty)$,
we have $x \le x^*$. By monotonicity of $p$, we get $p(x) \le p(x^*)$.
Thus, $r(p \mid \v, \infty) = p(x^*)$.
\end{proof}

\subsection{Piecewise-Linearization}
\label{sec:price-char:disc}

In this subsection, we show how to \emph{piecewise-linearize} a univariate convex pricing function,
i.e., we prove that a convex pricing function for a single dataset can be replaced by
a piecewise linear pricing function without losing revenue.

Convexity makes it easy to modify the pricing function without losing revenue.
In \cref{ex:convex-price-derivative}, if we increase the price such that
the sign of $\v'(x) - p'(x)$ does not invert for any $x$,
then the global maximum of $\v-p$ would remain the same.
Hence, the buyer would continue to either exhaust her budget
or buy at the same global maximum, which now has a higher price.

If we have multiple buyers, and each buyer $i$ has a linear valuation function $\v_i$,
then we round up the value of $p'(x)$ for each $x$ to the nearest element in the set $\{\v_i(1): i \in [n]\}$
(and round it down to $\max_i \v_i(1)$ if $p'(x) > \max_i \v_i(1)$).
Then $p'$ becomes piecewise constant, and so, $p$ becomes piecewise linear.
See \cref{fig:plin} in \cref{sec:overview:plin} for an example.

Note that such a modification does not always work for non-convex functions.
In \cref{ex:non-convex-price}, if we increase the per-unit price in the interval $(0.4, 0.6]$
from $3.5$ to $5$, then the loss incurred by purchasing the segment $(0.4, 0.6]$
will not be offset by the gain incurred by purchasing $(0.6, 1]$.
For large budgets, this would change the demand set from $\{1\}$ to $\{0.4\}$,
decreasing the revenue from 1.5 to 0.4.

We now formally define piecewise-linearization and prove that it does not decrease revenue.

\begin{definition}[Piecewise-linear functions]
\label{defn:plin}
Let $k \in \mathbb{N}$, $\zvec \in [0, 1]^{k-1}$, and $\alphavec \in \mathbb{R}^k$ such that
$z_1 \le \ldots \le z_{k-1}$ and $\alpha_1 \le \ldots \le \alpha_k$. Let $z_0 \defeq 0$.
We define $\plin(\zvec, \alphavec)$ to be the function $f: [0, 1] \to \mathbb{R}$
where $f(0) = 0$, and for any $x > 0$,
\[ f(x) \defeq \left(\sum_{j=1}^{i-1} \alpha_j(z_j - z_{j-1})\right) + \alpha_i(x - z_{i-1}), \]
where $i \in [k]$ is the largest such that $z_{i-1} < x$.
\end{definition}

\begin{remark}
\label{thm:plin-is-mcc}
For non-decreasing sequences $\zvec \in [0, 1]^{k-1}$ and $\alphavec \in \mathbb{R}^k$,
the function $f \defeq \plin(\zvec, \alphavec)$ is continuous and convex.
Moreover, if $\alpha_1 \ge 0$, then $f$ is also monotone.
\end{remark}

\begin{definition}[Piecewise-linearization]
\label{defn:disc}
Let $f: [0, 1] \to \mathbb{R}$ be a continuous convex function.
Let $S = \{\alpha_1, \ldots, \alpha_k\}$, where $\alpha_1 < \ldots < \alpha_k$.
For all $i \in [k-1]$, let $z_i \defeq \ell(f \mid \alpha_i)$ (c.f.~\cref{sec:price-char:rate-eq}).
Then $0 \le z_1 \le \ldots \le z_{k-1} \le 1$.
The function $\fhat \defeq f(0) + \plin(\zvec, \alphavec)$
is called the \emph{piecewise-linearization} or \emph{discretization} of $f$,
and is denoted by $\disc(f, S)$.
\end{definition}

To see how piecewise linearizing a pricing function $p$ affects the revenue,
we must first understand how the revenue-maximizing point in the demand set changes.
By \cref{thm:convex-price-rev}, we know that this point is $\ell(p \mid \v(1))$,
so let's first see how this point is affected.

\begin{lemma}
\label{thm:disc-dominates}
Let $f: [0, 1] \to \mathbb{R}$ be a continuous convex function,
let $S \subset \mathbb{R}_{\ge 0}$ be a finite set of non-negative numbers,
and let $\fhat \defeq \disc(f, S)$.
Let $\beta \in S$, let $x^* \defeq \ell(f \mid \beta)$, and let $\xhat \defeq \ell(\fhat \mid \beta)$.
Then $x^* \le \xhat$ and $f(x^*) \le \fhat(\xhat)$.
\end{lemma}
\begin{proof}
Let $S = \{\alpha_1, \ldots, \alpha_k\}$, where $\alpha_1 < \ldots < \alpha_k$.
By \cref{defn:disc}, $\fhat = \plin(\zvec, \alphavec)$ for some $\zvec \in [0, 1]^{k-1}$,
where $z_i = \ell(f \mid \alpha_i)$.
Let $z_0 = 0$, $z_k \defeq \ell(f \mid \alpha_i)$, and $\alpha_0 = 0$.

For all $i \in [k]$ and $x \in (z_{i-1}, z_i]$,
we have $f'(x) \in (\alpha_{i-1}, \alpha_i]$ and $\fhat'(x) = \alpha_i$.
Hence, for all $i \in [k-1]$, we have
$\ell(\fhat \mid \alpha_i) = z_i = \ell(f \mid \alpha_i)$.
Moreover, $\ell(\fhat \mid \alpha_k) = 1 \ge z_k = \ell(f \mid \alpha_k)$.

For all $i \in [k]$, we have $f'(z_i) \le \alpha_i$.
By \cref{thm:convex-ld-rates} in \cref{sec:real-analysis:ld},
we have $f(b) - f(a) \le f'(b)(b-a)$ for a convex function $f$ and $0 \le a < b \le 1$. Hence,
\[ f(z_i) - f(z_{i-1}) \le f'(z_i)(z_i - z_{i-1}) \le \alpha_i(z_i - z_{i-1}) = \fhat(z_i) - \fhat(z_{i-1}). \]
Summing over $i$, we get $f(z_i) \le \fhat(z_i)$ for all $i \in [k]$.
Moreover, since $\fhat$ is monotone, we get $\fhat(1) \ge \fhat(z_k) = f(z_k)$.

Hence, for all $i \in [k]$, we get
$\ell(\fhat \mid \alpha_i) \ge \ell(f \mid \alpha_i)$
and $\fhat(\ell(\fhat \mid \alpha_i)) \ge f(\ell(f \mid \alpha_i))$.
\end{proof}

\thmDiscRevenue*
\begin{proof}
$\phat$ is monotone, continuous, convex, and piecewise-linear by \cref{thm:plin-is-mcc}.
Let $x^* \defeq \ell(p \mid \v(1))$ and $\xhat \defeq \ell(\phat \mid \v(1))$.
By \cref{thm:disc-dominates}, we get $\phat(\xhat) \ge p(x^*)$.
By \cref{thm:convex-price-rev,thm:budget-exhaustion}, we get
$r(\phat \mid \v, b) = \min(b, \phat(\xhat)) \ge \min(b, p(x^*)) = r(p \mid \v, b)$.
\end{proof}

\subsection{Separability}
\label{sec:price-char:sep}

So far we have seen how to \emph{convexify} and \emph{piecewise-linearize} pricing functions.
However, for a separable pricing function, our convexification procedure operates on the entire function,
whereas piecewise-linearization operates on the components.
To compose these operations, we must understand how operations on a separable function
affect the components and vice-versa.
Hence, in this subsection, we explore the properties of separable functions.

\begin{remark}
Let $f: [0, 1]^m \to \mathbb{R}$ be a separable function
and $f_1, \ldots, f_m$ be its components.
Then for any $j \in [m]$ and any $x \in [0, 1]$, we have $f_j(x) = f_j(0) + f(x\vecE^{(j)}) - f(\vecZero)$.
Hence, the decomposition into components is unique up to the value at zero.
\end{remark}

A useful idea for analyzing separable functions is to induct on the number of components.
Say we have a separable function $h: [0, 1]^{m+n} \to \mathbb{R}$.
Let's denote the sum of the first $m$ components by $f$ and the last $n$ components by $g$.
Then $f$ and $g$ are separable, and we get $h(\xvec, \yvec) = f(\xvec) + g(\yvec)$
for all $\xvec \in [0, 1]^m$ and $\yvec \in [0, 1]^n$.

We begin by showing that a separable function being
monotone, (lower) continuous, or convex is equivalent to its components having the same properties.

\begin{restatable}{lemma}{thmSepProps}
\label{thm:sep-props}
Let $h(\xvec, \yvec) \defeq f(\xvec) + g(\yvec)$ for all $\xvec$ and $\yvec$, where $f: [0, 1]^m \to \mathbb{R}$,
$g: [0, 1]^n \to \mathbb{R}$, and $h: [0, 1]^{m+n} \to \mathbb{R}$. Then
\begin{tightenum}
\item $h$ is monotone iff $f$ and $g$ are monotone.
\item $h$ is lower-continuous iff $f$ and $g$ are lower-continuous.
\item $h$ is upper-continuous iff $f$ and $g$ are upper-continuous.
\item $h$ is convex iff $f$ and $g$ are convex.
\end{tightenum}
\end{restatable}
\begin{proof}
The proof is deferred (page~\pageref{prf:sep-props}) to \cref{sec:real-analysis:sep}.
\end{proof}

Next, we show that taking the convex hull of a separable function
is equivalent to taking the convex hull of each of its components.

\begin{restatable}{lemma}{thmConvexEnvSep}
\label{thm:convex-env-sep}
Let $h(\xvec, \yvec) \defeq f(\xvec) + g(\yvec)$ for all $\xvec$ and $\yvec$, where $f: [0, 1]^m \to \mathbb{R}$,
$g: [0, 1]^n \to \mathbb{R}$, and $h: [0, 1]^{m+n} \to \mathbb{R}$.
Let $\hhat \defeq \conv(h)$, $\fhat \defeq \conv(f)$, and $\ghat \defeq \conv(g)$.
Then $\hhat(\xvec, \yvec) = \fhat(\xvec) + \ghat(\yvec)$ for all $\xvec$ and $\yvec$.
\end{restatable}
\begin{proof}
For any $S, T \subseteq \mathbb{R}$, let $S + T \defeq \{x+y: x \in S, y \in T\}$.
Then $\inf(S + T) = \inf(S) + \inf(T)$.
We will show that for all $\xvec \in [0, 1]^m$ and $\yvec \in [0, 1]^n$, we have
$U_h((\xvec, \yvec)) = U_f(\xvec) + U_g(\yvec)$. Then
$\hhat(\xvec, \yvec) = \inf(U_h((\xvec, \yvec))) = \inf(U_f(\xvec)) + \inf(U_g(\yvec)) = \fhat(\xvec) + \ghat(\yvec)$.

Let $\xvechat \in [0, 1]^m$ and $\yvechat \in [0, 1]^n$.
Let $z_1 \in U_f(\xvechat)$ and $z_2 \in U_g(\yvechat)$.
This means that
$\exists p \in \mathbb{N}$, $\exists \mu \in \Delta_p$,
and $\exists \xvec^{(1)}, \ldots, \xvec^{(p)} \in [0, 1]^m$ such that
\begin{align*}
\xvechat &= \sum_{i=1}^p \mu_i \xvec^{(i)},
& z_1 &\ge \sum_{i=1}^p \mu_i f(\xvec^{(i)}),
\end{align*}
and $\exists q \in \mathbb{N}$, $\exists \nu \in \Delta_q$,
and $\exists \xvec^{(1)}, \ldots, \xvec^{(q)} \in [0, 1]^n$ such that
\begin{align*}
\yvechat &= \sum_{j=1}^q \nu_j \yvec^{(j)},
& z_2 &\ge \sum_{j=1}^q \nu_j g(\yvec^{(j)}).
\end{align*}

We will now show that $z_1 + z_2 \in U_h(\xvechat, \yvechat)$.
Let $\lambda_{i,j} \defeq \mu_i\nu_j$ for all $i \in [p]$ and $j \in [q]$. Then
\begin{align*}
& \sum_{i=1}^p \sum_{j=1}^q \lambda_{i,j} = 1,
\qquad \sum_{i=1}^p \sum_{j=1}^q \lambda_{i,j}\xvec^{(i)} = \xvechat,
\qquad \sum_{i=1}^p \sum_{j=1}^q \lambda_{i,j}\yvec^{(j)} = \yvechat,
\\ &\sum_{i=1}^p \sum_{j=1}^q \lambda_{i,j}h(\xvec^{(i)}, \yvec^{(j)})
    = \sum_{i=1}^p \mu_i f(\xvec^{(i)}) + \sum_{j=1}^q \nu_j g(\yvec^{(j)})
    \le z_1 + z_2.
\end{align*}
Hence, $z_1 + z_2 \in U_h(\xvechat, \yvechat)$.

Suppose $\zhat \in U_h(\xvechat, \yvechat)$.
This means that $\exists k \in \mathbb{N}$, $\exists \rho \in \Delta_k$,
$\exists \xvec^{(1)}, \ldots, \xvec^{(k)} \in [0, 1]^m$,
and $\exists \yvec^{(1)}, \ldots, \yvec^{(k)} \in [0, 1]^n$ such that
\begin{align*}
\xvechat &= \sum_{i=1}^k \rho_i \xvec^{(i)},
\qquad \yvechat = \sum_{i=1}^k \rho_i \yvec^{(i)},
\\ \zhat &\ge \sum_{i=1}^k \rho_i (f(\xvec^{(i)}) + g(\yvec^{(j)})).
\end{align*}
Let $z' \defeq \sum_{i=1}^k f(\xvec^{(i)})$.
Then $z' \in U_f(\xvechat)$ and $\zhat - z' \in U_g(\yvechat)$.
Hence, $\zhat = z' + (\zhat - z') \in U_f(\xvechat) + U_g(\yvechat)$.
Since we picked $z_1$, $z_2$, and $\zhat$ arbitrarily,
we get $U_h(\xvechat, \yvechat) = U_f(\xvechat) + U_g(\yvechat)$.
This completes the proof of the lemma.
\end{proof}

We now show (in \cref{thm:sep-opt}) that optimizing over a separable function
is equivalent to optimizing over its components.
This helps us relate the demand sets and revenue of a separable pricing function
to the demand sets and revenues of its components (see \cref{thm:revenue-sep}).

\begin{restatable}{lemma}{thmSepOpt}
\label{thm:sep-opt}
Let $h(\xvec, \yvec) \defeq f(\xvec) + g(\yvec)$ for all $\xvec$ and $\yvec$,
where $f: [0, 1]^m \to \mathbb{R}$, $g: [0, 1]^n \to \mathbb{R}$, and $h: [0, 1]^{m+n} \to \mathbb{R}$.
Let $X^* \defeq \argmax_{\xvec \in [0, 1]^m} f(\xvec)$, $Y^* \defeq \argmax_{\yvec \in [0, 1]^n} g(\yvec)$,
and $Z^* \defeq \argmax_{\zvec \in [0, 1]^{m+n}} h(\zvec)$. Then $Z^* = X^* \times Y^*$.
\end{restatable}
\begin{proof}
The proof is deferred (page~\pageref{prf:sep-opt}) to \cref{sec:real-analysis:sep}.
\end{proof}

\begin{lemma}
\label{thm:revenue-sep}
For $j \in [m]$, let $p_j: [0, 1] \to \mathbb{R}$ be an MLC pricing function.
For $\xvec \in [0, 1]^m$, let $\p(\xvec) \defeq \sum_{j=1}^m p_j(x_j)$.
Let $u_j: [0, 1] \to \mathbb{R}_{\ge 0}$ be a continuous valuation function.
For $\xvec \in [0, 1]^m$, let $\v(\xvec) \defeq \sum_{j=1}^m u_j(x_j)$. Then
\[ r(\p \mid \v, \infty) = \sum_{j=1}^m r(p_j \mid u_j, \infty). \]
\end{lemma}
\begin{proof}
By applying \cref{thm:sep-opt} on $\v-p$, we get that $\Dem(\p) = \prod_{j=1}^m \Dem(p_j)$.

Since the supremum of a sum of functions of independent variables
equals the sum of their suprema (\cref{thm:sup-decomp} in \cref{sec:ext-real}), we get
\[ r(\p) = \sup_{\xvec \in \Dem(\p)} \p(\xvec) = \sup_{x_j \in \Dem(p_j) \forall j \in [m]} \sum_{j=1}^m p_j(x_j)
    = \sum_{j=1}^m \sup_{x_j \in \Dem(p_j)} p_j(x_j) = \sum_{j=1}^m r(p_j).
\qedhere \]
\end{proof}

\subsection{Proof of the Structural Theorem}
\label{sec:price-char:together}

\begin{theorem}
\label{thm:convex-disc-dom}
Let there be $n$ buyers and $m$ datasets.
Each buyer $i \in [n]$ has a budget $b_i \in \mathbb{R}_{\ge 0}$.
Let $\v_{i,j} \in \mathbb{R}_{\ge 0}$ be buyer $i$'s value for dataset $j$.
Buyers' valuations are linear, i.e., any buyer $i$'s valuation for a bundle $\xvec \in [0, 1]^m$
is given by $\v_i(\xvec) \defeq \sum_{j=1}^m \v_{i,j}x_j$.
For any $j \in [m]$, let $S_j \defeq \{\v_{i,j}: i \in [n]\}$
be the set of distinct valuations for dataset $j$.

For each dataset $j \in [m]$, let $p_j: [0, 1] \to \mathbb{R}_{\ge 0}$ be an MLC pricing function.
Let $\phat_j \defeq \disc(\conv(p_j), S_j)$ for all $j \in [m]$.
For all $\xvec \in [0, 1]^m$, let $\p(\xvec) \defeq \sum_{j=1}^m p_j(x_j)$
and $\phat(\xvec) \defeq \sum_{j=1}^m \phat_j(x_j)$.
Then for each buyer, the revenue from $\pvechat$ is not less than that from $\p$, i.e.,
$r(\pvechat \mid \v_i, b_i) \ge r(\p \mid \v_i, b_i)$ for all $i \in [n]$.
Moreover, $\phat_j$ is monotone, continuous, convex, and piecewise-linear for all $j \in [m]$.
\end{theorem}
\begin{proof}
By \cref{thm:sep-props}, $\p$ is MLC. Let $\ptild_j \defeq \conv(p_j)$.
By \cref{thm:convex-env-props},
$\ptild_j$ is monotone, continuous, convex, and $\ptild_j(0) = 0$.
Let $\pvectild(\xvec) \defeq \sum_{j=1}^m \ptild_j(x_j)$ for all $j \in [m]$.
By \cref{thm:convex-env-sep}, we get that $\pvectild = \conv(\p)$.
By \cref{thm:convex-env-dominates}, for all $i \in [n]$, we get
$r(\pvectild \mid \v_i, b_i) \ge r(\p \mid \v_i, b_i)$.

Since $\phat_j \defeq \disc(\ptild_j, S_j)$, we get that
$\phat_j = \plin(\zvec^{(j)}, S_j)$ for some non-decreasing $\zvec^{(j)} \in [0, 1]^{|S_j|}$
(we treat $S_j$ as a non-decreasing sequence consisting of values $\{\v_{i,j}: i \in [n]\}$).
By \cref{thm:plin-is-mcc}, we get that $\phat_j$ is monotone, continuous, and convex.
By \cref{thm:sep-props}, we get that $\pvechat$ is also monotone, continuous, and convex.
We will show that for all $i \in [n]$, we have $r(\pvechat \mid \v_i, b_i) \ge r(\pvectild \mid \v_i, b_i)$.

By \cref{thm:disc-revenue}, for all $j \in [m]$, we get
$r(\phat_j \mid \v_{i,j}, \infty) \ge r(\ptild_j \mid \v_{i,j}, \infty)$.
By \cref{thm:revenue-sep}, we get that
\[ r(\pvechat \mid \v_i, \infty)
    = \sum_{j=1}^m r(\phat_j \mid \v_{i,j}, \infty)
    \ge \sum_{j=1}^m r(\ptild_j \mid \v_{i,j}, \infty)
    = r(\pvectild \mid \v_i, \infty). \]
By \cref{thm:budget-exhaustion}, we get
$r(\pvechat \mid \v_i, b_i) = \min(b_i, r(\pvechat \mid \v_i, \infty))
\ge \min(b_i, r(\pvectild \mid \v_i, \infty)) = r(\pvectild \mid \v_i, b_i)$.

Hence, for all $i \in [n]$, we get
$r(\pvechat \mid \v_i, b_i) \ge r(\pvectild \mid \v_i, b_i) \ge r(\p \mid \v_i, b_i)$.
\end{proof}

The structural theorem is a simple corollary of the above result.

\thmStruct*
\begin{proof}
By \cref{thm:sep-props}, each component of $\p$ is MLC. Now apply \cref{thm:convex-disc-dom}.
\end{proof}

\subsection{The Sharding Interpretation}
\label{sec:sharding-lp}

By the structural theorem (\cref{thm:struct}), we can restrict our attention to
monotone, continuous, convex, and piecewise-linear pricing functions.
We now give an insightful interpretation of such pricing functions.

\begin{example}
\label{ex:sharding}
Consider a single dataset whose pricing function is
$p \defeq \plin((0.4, 0.4),\allowbreak (10, 25, 65))$, i.e.,
\begin{align*}
p'(x) &= \begin{cases}
10 & \text{ if } 0 \le x \le 0.4
\\ 25 & \text{ if } 0.4 \le x \le 0.8
\\ 65 & \text{ if } 0.8 \le x \le 1
\end{cases},
& p(x) &= \begin{cases}
10x & \text{ if } 0 \le x \le 0.4
\\ 4 + 25(x-0.4) & \text{ if } 0.4 \le x \le 0.8
\\ 14 + 65(x-0.8) & \text{ if } 0.8 \le x \le 1
\end{cases}.
\end{align*}
See \cref{fig:plc-ex} (in \cref{sec:our-contrib}) for plots of $p$ and $p'$.

Instead of viewing this example as a single dataset having a non-linear pricing function,
we can think of the dataset as being broken into three \emph{shards},
each with a different linear pricing function:
\begin{tightemize}
\item Shard 1: size $0.4$, per-unit price $10$.
\item Shard 2: size $0.4$, per-unit price $25$.
\item Shard 3: size $0.2$, per-unit price $65$.
\end{tightemize}

Giving the buyer a $0.6$ fraction of the dataset means
giving her the first shard and half of the second shard.

Instead of presenting to the buyers a single dataset having pricing function $p$,
suppose we present each shard as a separate dataset.
This means we allow buyers to purchase the shards \emph{out-of-order}, e.g.,
allowing them to purchase part of shard 2 without first purchasing the entire shard 1.
However, one can show that buyers will not utilize this extra freedom:
buyers have the same per-unit value for the shards,
so they will first purchase the entire shard with the lowest price
before moving on to the next shard.
\end{example}

In general, for non-decreasing sequences $\zvec \in [0, 1]^{k-1}$ and $\alphavec \in \mathbb{R}_{\ge 0}^k$,
pricing a dataset using the function $p_j \defeq \plin(\zvec, \alphavec)$
is the same as breaking the dataset into $k$ shards,
where the $t\Th$ shard has size $z_t$ and per-unit price $\alpha_t$.

Once we find the revenue-maximizing pricing by solving LP \cref{eqn:lp},
we can also find the corresponding allocation.
For each buyer $i \in [n]$, we would like to find a bundle $\xvec_i \in [0, 1]^m$ that maximizes her net utility,
and subject to that, maximizes the amount she pays.
Since each shard has linear pricing, the problem reduces to fractional knapsack over the shards,
which can be solved efficiently.

\subsection{Number of Kinks in LP}
\label{sec:price-char:kinks}

In \cref{sec:overview:lp}, we showed that in an optimal basic feasible solution,
the total number of \emph{kinks} (points of non-differentiability)
in the PLC pricing functions is at most $n$.
We now present an example with $n-1$ kinks, showing that our upper bound
on the number of kinks is nearly tight.

\thmMaxKinks*
\begin{proof}
One can easily verify that $(\bvec, \zvec^*)$ is a feasible solution to LP \eqref{eqn:lp}.
Moreover, it is an optimal solution because every buyer exhausts her budget.

Let $(\rvec, \zvec)$ be an optimal solution to LP \eqref{eqn:lp}.
Then we have $\rvec = \bvec$. We will now show that $\zvec = \zvec^*$.
By the feasibility constraints, for each buyer $i \in [n]$, we have
\[ \sum_{j=1}^n jz_{j,1} \ge b_i. \]
Let
\[ a_i \defeq \frac{1}{i(i+1)} = \frac{1}{i} - \frac{1}{i+1} \]
for all $i \in [n-1]$, and let $a_n \defeq 1/n$.
Then taking a sum of the above constraints, weighted by $(a_1, \ldots, a_n)$, we get
\begin{align*}
& \sum_{i=1}^n a_i \left(\sum_{j=1}^n jz_{j,1}\right) \ge \sum_{i=1}^n a_ib_i
\\ &\iff \sum_{j=1}^n jz_{j,1} \left(\sum_{i=j}^n a_j\right)
    \ge \sum_{i=1}^{n-1} \frac{1}{i(i+1)}\cdot\frac{i(i+1)}{2n} + \frac{1}{n}\cdot\frac{n(n+1)}{2n}
\\ &\iff \sum_{j=1}^n z_{j,1} \ge 1.
\end{align*}
Thus, if any of the constraints is not tight, then we get $\sum_{j=1}^n z_{j,1} > 1$,
which is a contradiction. Hence, all constraints must be tight.
One can verify that $z^*$ is the unique solution to this system of equations.
\end{proof}

\section{Comparison to Ironing Techniques}
\label{sec:ironing-review}

Our analysis of the optimal MLC pricing function can also be interpreted as an ``ironing'' process,
which has also been considered in the literature of auction theory~\citep{myerson1981optimal}.
In that context, ironing refers to the procedure of transforming a non-monotonic virtual value or allocation function into a monotonic one, thereby ensuring incentive compatibility.
Our procedure in \cref{sec:rev_max} is conceptually similar to this classical ironing approach, but it differs in several important details.
First, Myerson's mechanism considers the auction of \emph{a single item}, and the ironing process typically only works on the virtual value or allocation function of a \emph{single bidder}.
In contrast, in our setting, we consider the general pricing function of $m$ datasets for $n$ buyers, and any change to it would affect the revenue gained from each buyer.
Second, although Myerson's mechanism is computationally efficient for one item and multiple bidders, extending it to the settings of multiple-item auction is known to be difficult~\citep{daskalakis2014complexity}.
In contrast, our pricing strategy works for multiple datasets and can be captured by a linear program.
Furthermore, one limitation of Myerson's mechanism is that it may not always be easy to implement, since the payment function is determined by a complex integral expression.
In contrast, our pricing function is well-structured: piecewise-linear, and the derivatives are only taken from buyers' values, which makes it easy to implement in practical settings.

\section{Details on Real and Convex Analysis}
\label{sec:real-analysis}

We present several results from real and convex analysis here.
Most of these results are standard and well-known.
\begin{tightenum}
\item \cref{sec:ext-real} defines the extended real number system $\mathbb{R} \cup \{-\infty, \infty\}$
    and extends $\sup$ and $\inf$ for it.
\item \cref{sec:real-analysis:seq-conv} is on convergence of sequences.
\item \cref{sec:real-analysis:mlc} studies monotone and continuous functions.
\item \cref{sec:real-analysis:convex-sets} is on convex sets.
\item \cref{sec:real-analysis:convex-functions} is on convex functions.
\item \cref{sec:real-analysis:ld} studies the left derivative of univariate convex functions.
\item \cref{sec:real-analysis:sep} is on separable functions.
\end{tightenum}

\subsection{Extended Real Numbers}
\label{sec:ext-real}

In this paper, we will often use the extended real number system $\mathbb{R} \cup \{-\infty, \infty\}$.
We assume that
\begin{tightenum}
\item $-\infty < x < \infty$ for all $x \in \mathbb{R}$.
\item $x + \infty = \infty + x = \infty$ for all $x \in \mathbb{R} \cup \{\infty\}$.
\item $x - \infty = - \infty + x = -\infty$ for all $x \in \mathbb{R} \cup \{-\infty\}$.
\end{tightenum}

Let $X \subseteq \mathbb{R} \cup \{-\infty, \infty\}$.
$\sup(X)$ is defined to be the least upper bound on $X$, i.e.,
\begin{enumerate}
\item $\sup(X) \defeq \infty$ if $\forall M \in \mathbb{R}$, $\exists x \in X$ such that $x > M$.
\item $\sup(X) \defeq -\infty$ if $X \subseteq \{-\infty\}$.
\item $\sup(X) \in \mathbb{R}$ if both of the following hold:
    \begin{tightenum}
    \item $\forall x \in X$, we have $x \le \sup(X)$.
    \item $\forall \eps > 0$, $\exists x \in X$ such that $\sup(X) - \eps < x$.
    \end{tightenum}
\end{enumerate}

$\inf(X)$ is the greatest lower bound on $X$, and is defined analogously.
Also, we have $\inf(X) \defeq -\sup(\{-x: x \in X\})$.

\begin{observation}
\label{thm:sup-inf-union}
Let $f: S \cup T \to \mathbb{R}$. Then
\[ \inf_{x \in S \cup T} f(x) = \min\left(\inf_{x \in S} f(x), \inf_{x \in T} f(x)\right), \]
\[ \sup_{x \in S \cup T} f(x) = \max\left(\sup_{x \in S} f(x), \sup_{x \in T} f(x)\right). \]
\end{observation}

\begin{lemma}
\label{thm:sup-sum}
Let $f, g: X \to \mathbb{R}$. Then
\[ \sup_{x \in X} (f(x) + g(x)) \le \sup_{x \in X} f(x) + \sup_{x \in X} g(x). \]
\end{lemma}
\begin{proof}
Let $\alpha \defeq \sup_{x \in X} f(x)$ and $\beta \defeq \sup_{x \in X} g(x)$.
Then for any $x \in X$, we have $f(x) + g(x) \le \alpha + \beta$,
so $\sup_{x \in X} (f(x) + g(x)) \le \alpha + \beta$.
\end{proof}

\begin{lemma}
\label{thm:sup-decomp}
Let $X$ and $Y$ be non-empty sets. Let $f: X \to \mathbb{R}$ and $g: Y \to \mathbb{R}$. Then
\[ \sup_{(x, y) \in X \times Y} (f(x) + g(y)) = \sup_{x \in X} f(x) + \sup_{y \in Y} g(y). \]
\end{lemma}
\begin{proof}
Let $\alpha \defeq \sup_{x \in X} f(x)$, $\beta \defeq \sup_{y \in Y} g(y)$,
and $\gamma \defeq \sup_{(x, y) \in X \times Y} (f(x) + g(y))$.
Then $\gamma \le \alpha + \beta$ by \cref{thm:sup-sum}.

Suppose $\alpha = \infty$. Let $y \in Y$.
For all $M \in \mathbb{R}_{\ge 0}$, $\exists x \in X$ such that $f(x) > M$.
Then $\gamma \ge f(x) + g(y) > M + g(y)$.
By picking an arbitrarily large $M$, we can make $\gamma$ arbitrarily large, so $\gamma = \infty$.
Similarly, if $\beta = \infty$, then $\gamma = \infty$.
Now suppose $\alpha$ and $\beta$ are both finite.

For any $\eps > 0$, there exists $x \in X$ such that $f(x) > \alpha - \eps$
and $\exists y \in Y$ such that $g(y) > \beta - \eps$.
Hence, $\gamma \ge f(x) + g(y) \ge \alpha + \beta - 2\eps$.
By making $\eps$ arbitrarily small, we get $\gamma \ge \alpha + \beta$.
Hence, $\gamma = \alpha + \beta$.
\end{proof}

\subsection{Convergence of Sequences}
\label{sec:real-analysis:seq-conv}

\begin{definition}
\label{defn:seq-convergence}
For all $i \in \mathbb{N}$, let $x^{(i)} \in \mathbb{R}^m$.
The infinite sequence $X \defeq (x^{(i)})_{i \in \mathbb{N}}$ is said to
\emph{converge} to a point $\xhat \in \mathbb{R}^m$ if
for all $\eps > 0$, there exists $n \in \mathbb{N}$ such that
for all $i \ge n$, we have $\linf{x^{(i)} - \xhat} < \eps$.
(A sequence converges to a unique point, as proved in \cref{thm:convergence-unique}.)
\end{definition}

Most results here are standard, so their proofs are ommitted.

\begin{lemma}
\label{thm:convergence-unique}
If the sequence $X \defeq (x^{(i)})_{i \in \mathbb{N}}$ of points in $\mathbb{R}^m$
converges to $\xhat \in \mathbb{R}^m$ and also converges to $x^* \in \mathbb{R}^m$, then $\xhat = x^*$.
\end{lemma}

\begin{lemma}
\label{thm:convergence-scaling}
Let $c \in \mathbb{R}$. If $X \defeq (x^{(i)})_{i \in \mathbb{N}}$ converges to $x^* \in \mathbb{R}^m$,
then $cX \defeq (cx^{(i)})_{i \in \mathbb{N}}$ converges to $cx^*$.
\end{lemma}

\begin{lemma}
\label{thm:convergence-sum}
If $X \defeq (x^{(i)})_{i \in \mathbb{N}}$ converges to $x^* \in \mathbb{R}^m$,
and $Y \defeq (y^{(i)})_{i \in \mathbb{N}}$ converges to $y^* \in \mathbb{R}^m$,
then $X+Y \defeq (x^{(i)} + y^{(i)})_{i \in \mathbb{N}}$ converges to $x^* + y^*$.
\end{lemma}

\begin{lemma}
\label{thm:convergence-decomp}
Let $X \defeq (x^{(i)})_{i \in \mathbb{N}}$, $Y \defeq (y^{(i)})_{i \in \mathbb{N}}$,
and $Z \defeq ((x^{(i)}, y^{(i)}))_{i \in \mathbb{N}}$ be sequences of points in
$\mathbb{R}^m$, $\mathbb{R}^n$, and $\mathbb{R}^{m+n}$, respectively.
Let $x^* \in \mathbb{R}^m$ and $y^* \in \mathbb{R}^n$.
Then $Z$ converges to $(x^*, y^*)$ iff $X$ converges to $x^*$ and $Y$ converges to $Y^*$.
\end{lemma}
\begin{proof}
\begin{align*}
& Z \text{ converges to } (x^*, y^*)
\\ &\iff \forall \delta > 0, \exists n \in \mathbb{N}, \forall i \ge n,
    \linf{(x^{(i)}, y^{(i)}) - (x^*, y^*)} < \delta
\\ &\iff \forall \delta > 0, \exists n \in \mathbb{N}, \forall i \ge n,
    \max(\linf{x^{(i)} - x^*}, \linf{y^{(i)} - y^*}) < \delta
\\ &\iff X \text{ converges to } x^* \text{ and } Y \text{ converges to } y^*.
\qedhere
\end{align*}
\end{proof}

\begin{definition}
\label{defn:bounded}
A set $S \subseteq \mathbb{R}^m$ is \emph{bounded} if there exists $z \in \mathbb{R}_{\ge 0}$
such that for all $x \in S$, we have $\linf{x} \le z$.
\end{definition}

\begin{definition}
\label{defn:closed}
A set $S \subseteq \mathbb{R}^m$ is \emph{closed} if for every sequence
$(x^{(i)})_{i \in \mathbb{N}}$ of points in $S$ that converges to some $\xhat \in \mathbb{R}^m$,
we have $\xhat \in S$.
\end{definition}

\begin{theorem}[Bolzano-Weierstrass]
\label{thm:bolzano-weier}
(Every bounded sequence has a convergent subsequence.)
Let $S \subseteq \mathbb{R}^m$ be a bounded set.
Then for any infinite sequence $(x^{(i)})_{i \in \mathbb{N}}$ of points in $S$,
there exists a subset $I \subseteq \mathbb{N}$ such that
the subsequence $(x^{(i)})_{i \in I}$ converges to some point $\xhat \in \mathbb{R}^m$.
\end{theorem}

\subsection{Monotonicity and Continuity}
\label{sec:real-analysis:mlc}

\begin{restatable}[monotone function]{definition}{defnMonotone}
\label{defn:monotone}
For any two vectors $\xvec, \yvec \in [0, 1]^m$, we say that
$\xvec \preceq \yvec$ iff $x_j \le y_j$ for all $j \in [m]$.
A function $f: [0, 1]^m \to \mathbb{R} \cup \{-\infty, \infty\}$ is called \emph{monotone} if
for all $\xvec, \yvec \in [0, 1]^m$ such that $\xvec \preceq \yvec$, we have $f(\xvec) \le f(\yvec)$.
\end{restatable}

Note how we allow $f$'s co-domain to include $\infty$ and $-\infty$ in \cref{defn:monotone}.
This will be useful in \cref{sec:price-char:rate-eq}.
See \cref{sec:ext-real} for our conventions regarding infinity, $\sup$, and $\inf$.

\begin{definition}[continuity]
\label{defn:continuity}
Let $f: [0, 1]^m \to \mathbb{R} \cup \{-\infty, \infty\}$ be a function.
For any $\xvechat \in [0, 1]^m$ and $\delta > 0$, define the $\delta$-neighborhood of $\xvechat$
as $N_{\delta}(\xvechat) \defeq \{\xvec \in [0, 1]^m: \linf{\xvec - \xvechat} \le \delta\}$.
Define the \emph{lower limit} of $f$ at $\xvechat$ as
\[ f^-(\xvechat) \defeq \sup_{\delta > 0} \inf_{\xvec \in N_{\delta}(\xvechat)} f(\xvec). \]
Define the \emph{upper limit} of $f$ at $\xvechat$ as
\[ f^+(\xvechat) \defeq \inf_{\delta > 0} \sup_{\xvec \in N_{\delta}(\xvechat)} f(\xvec). \]
It is easy to see that $f^-(\xvechat) \le f(\xvechat) \le f^+(\xvechat)$, since $\xvechat \in N_{\delta}(\xvechat)$.
$f$ is called \emph{lower-continuous} at $\xvechat$ if $f^-(\xvechat) = f(\xvechat)$.
$f$ is called \emph{upper-continuous} at $\xvechat$ if $f^+(\xvechat) = f(\xvechat)$.
$f$ is called \emph{continuous} at $\xvechat$ if $f^+(\xvechat) = f^-(\xvechat)$.
$f$ is called MLC if $f$ is both monotone and lower-continuous.
\end{definition}

\begin{remark}
Traditionally, continuity is defined as follows.
$f: [0, 1]^m \to \mathbb{R}$ is continuous at $\xvechat$ if
$\forall \eps > 0$, $\exists \delta > 0$, $\forall \xvec \in N_{\delta}(\xvechat)$,
$|f(\xvec) - f(\xvechat)| \le \eps$.
One can show that this definition is equivalent to \cref{defn:continuity}.
\end{remark}

\begin{lemma}
\label{thm:monotone-continuous}
Let $f: [0, 1]^m \to \mathbb{R} \cup \{-\infty, \infty\}$ be a monotone function and let $\xhat \in [0, 1]^m$.
For any $\delta > 0$ and $j \in [m]$, let
let $\xhat^{[-\delta]}_j \defeq \max(0, \xhat_j - \delta)$
and $\xhat^{[+\delta]}_j \defeq \min(1, \xhat_j + \delta)$.
Then
\begin{align*}
f^-(\xhat) &\defeq \sup_{\delta > 0} f(\xhat^{[-\delta]}),
& f^+(\xhat) &\defeq \inf_{\delta > 0} f(\xhat^{[+\delta]}).
\end{align*}
\end{lemma}
\begin{proof}
By monotonicity, for all $x \in N_{\delta}(\xhat)$, we get
$f(\xhat^{[-\delta]}) \le f(x) \le f(\xhat^{[+\delta]})$.
\end{proof}

\begin{restatable}[monotone support]{lemma}{thmSeqMonotonization}
\label{thm:seq-monotonization}
Let $X \defeq (\xvec^{(i)})_{i \in \mathbb{N}}$ be an infinite sequence of points in $[0, 1]^m$.
For all $k \in \mathbb{N}$ and $j \in [m]$, define
$\xhat^{(k)}_j \defeq \inf_{i \ge k} x^{(i)}_j$.
Let $\Xhat \defeq (\xvechat^{(i)})_{i \in \mathbb{N}}$ be another sequence of points in $[0, 1]^m$.
For any $j \in [m]$, define $x^*_j \defeq \sup_{i \in \mathbb{N}} \xhat^{(i)}_j$. Then
\begin{tightenum}
\item $\Xhat$ is monotone, i.e., $\forall i \in \mathbb{N}$, we have $\xvechat^{(i)} \preceq \xvechat^{(i+1)}$.
\item $\Xhat$ is dominated by $X$, i.e., $\forall i \in \mathbb{N}$, we have $\xvechat^{(i)} \preceq \xvec^{(i)}$.
\item $\Xhat$ converges to $\xvec^*$.
\item If $X$ converges to a point, that point must be $\xvec^*$.
\item If $f: [0, 1]^m \to \mathbb{R} \cup \{-\infty, \infty\}$ is an MLC function, then
    $\sup_{i \in \mathbb{N}} f(\xvechat^{(i)}) = f(\xvec^*)$.
\end{tightenum}
\end{restatable}
\begin{proof}
For all $k \in \mathbb{N}$ and $j \in [m]$, we have
$\xhat^{(k)}_j = \min(x^{(k)}_j, \xhat^{(k+1)}_j)$,
so we get $\xhat^{(k)} \preceq x^{(k)}$ and $\xhat^{(k)} \preceq x^{(k+1)}$.
Hence, $\Xhat$ is monotone and dominated by $\Xhat$.

The convergence of $\Xhat$ to $x^*$ follows from \cref{thm:convergence-decomp}
and the monotone convergence theorem for sequences of real numbers.

Suppose $X$ converges to $\xtild$. Pick any $\eps > 0$ and $j \in [m]$.
Then $\exists n \in \mathbb{N}$ such that for all $i \ge n$, we have
$\xtild_j - \eps < x^{(i)}_j < \xtild_j + \eps$.
Hence, we also have $\xtild_j - \eps \le \xhat^{(i)}_j < \xtild_j + \eps$ for all $i \ge n$.
Therefore, $\Xhat$ also converges to $\xtild$, so $\xtild = x^*$ (by \cref{thm:convergence-unique}).

Let $y_i \defeq f(\xhat^{(i)})$ for all $i \in \mathbb{N}$.
Let $Y \defeq (y_i)_{i \in \mathbb{N}}$. We need to show that $\sup(Y) = f(x^*)$.
Since $\xhat^{(i)} \preceq x^*$, we get $\sup(Y) \le f(x^*)$.
Pick any $\delta > 0$. Since $\Xhat$ converges to $x^*$, there exists $i \in \mathbb{N}$
such that $\linf{\xhat^{(i)} - x^*} < \delta$. Hence,
\[ f({x^*}^{[-\delta]}) \le f(\xhat^{(i)}) \le \sup(Y). \]
Since this is true for all $\delta > 0$, we get $f^-(x^*) \le \sup(Y)$.
Since $f$ is MLC, we get $f(x^*) = f^-(x^*) = \sup(Y)$.
\end{proof}

\begin{lemma}
\label{thm:mlc-feas-compact}
Let $f: [0, 1]^m \to \mathbb{R} \cup \{-\infty, \infty\}$ be an MLC function.
Then for any $b \in \mathbb{R}$, the set $S \defeq \{\xvec \in [0, 1]^m: f(\xvec) \le b\}$
is either empty or closed.
\end{lemma}
\begin{proof}
If $S = \emptyset$, we are done. Now assume $S \neq \emptyset$.
$S \subseteq [0, 1]^m$, so $S$ is bounded.

Let $X \defeq (\xvec^{(i)})_{i \in \mathbb{N}}$ be a sequence of points in $S$
that converges to $\xvec^*$. We want to show that $\xvec^* \in S$, i.e., $f(\xvec^*) \le b$.
Construct the sequence $\Xhat$ as in \cref{thm:seq-monotonization}. Then
\[ f(\xvec^*) = \sup_{i \in \mathbb{N}} f(\xvechat^{(i)}) \le \sup_{i \in \mathbb{N}} f(\xvec^{(i)}) \le b. \]
Hence, $\xvec^* \in S$, so $S$ is closed.
\end{proof}

\lemMLCDemandNonempty*
\label{prf:mlc-demand-nonempty}
\begin{proof}
Let $u^*(p) \defeq \sup_{\xvec \in F(p)} (\v(\xvec) - p(\xvec))$.
Then there exists a sequence $X \defeq (\xvec^{(i)})_{i \in \mathbb{N}}$
such that for all $i \in \mathbb{N}$, we have $\xvec^{(i)} \in F(p)$ and
$u^*(p) - 2^{-i} < \v(\xvec^{(i)}) - p(\xvec^{(i)}) \le u^*(p)$.
By the Bolzano-Weierstrass Theorem (\cref{thm:bolzano-weier} in \cref{sec:real-analysis:seq-conv}),
we can assume \wLoG{} that $X$ converges to some point $\xvec^*$.
By \cref{thm:mlc-feas-compact}, $F(p)$ is closed, so $\xvec^* \in F(p)$.

The sequence $(\v(\xvec^{(i)}) - p(\xvec^{(i)}))_{i \in \mathbb{N}}$ converges to $u^*(p)$.
Since $\v$ is continuous, the sequence $(\v(\xvec^{(i)}))_{i \in \mathbb{N}}$ converges to $\v(\xvec^*)$.
By \cref{thm:convergence-scaling,thm:convergence-sum} (in \cref{sec:real-analysis:seq-conv}),
$(p(\xvec^{(i)}))_{i \in \mathbb{N}}$ converges to $\v(\xvec^*) - u^*(p)$.

Construct the sequence $\Xhat$ as in \cref{thm:seq-monotonization}.
Since $\p$ is MLC, we get that $(\p(\xvechat^{(i)}))_{i \in \mathbb{N}}$ converges to $\p(\xvec^*)$.
By \cref{thm:convergence-scaling,thm:convergence-sum} (in \cref{sec:real-analysis:seq-conv}),
the sequence $Q \defeq (\p(\xvec^{(i)}) - \p(\xvechat^{(i)}))_{i \in \mathbb{N}}$
converges to $\v(\xvec^*) - u^*(p) - \p(\xvec^*)$. Since all elements in $Q$ are non-negative,
$Q$ converges to a non-negative value. Hence, $\v(\xvec^*) - \p(\xvec^*) \ge u^*(\p)$.
\end{proof}

\begin{lemma}[intermediate value theorem]
\label{thm:intermed-value}
Let $f: [a, b] \to \mathbb{R}$ be a continuous function.
Then for all $y \in [\min(f(a), f(b)), \max(f(a), f(b))]$,
there exists $x \in [a, b]$ such that $f(x) = y$.
\end{lemma}

\subsection{Convex Sets}
\label{sec:real-analysis:convex-sets}

\begin{definition}
\label{defn:convex-set}
$S \subseteq \mathbb{R}^m$ is convex iff
for all $x, y \in S$ and $\alpha \in (0, 1)$, we have
$(1-\alpha)x + \alpha y \in S$.
\end{definition}

\begin{definition}[Convex hull of a set]
Let $S \subseteq \mathbb{R}^m$. The convex hull of $S$ is defined as
\[ \conv(S) \defeq \left\{\xvec \in \mathbb{R}^m: \begin{array}{ll}
    \exists k \in \mathbb{N}, \exists \lambda \in \Delta_k, \exists \xvec^{(1)}, \ldots, \xvec^{(k)} \in S
    \\ \text{such that } \xvec = \sum_{i=1}^k \lambda_i \xvec^{(i)}
    \end{array} \right\}. \]
\end{definition}

It is easy to see that for any $S \subseteq \mathbb{R}^m$, $\conv(S)$ is convex.

\begin{observation}
\label{thm:convex-hull-subset}
If $A \subseteq B \subseteq \mathbb{R}^m$, then $\conv(A) \subseteq \conv(B)$.
\end{observation}

\begin{observation}
\label{thm:convex-hull-cap}
Let $A, B \subseteq \mathbb{R}^m$.
Then $\conv(A \cap B) \subseteq \conv(A) \cap \conv(B)$.
\end{observation}

\begin{lemma}[Theorem 17.2 in \cite{rockafellar1997convex17}]
\label{thm:conv-closed}
If $S \subseteq \mathbb{R}^m$ is closed and bounded,
then $\conv(S)$ is closed and bounded.
\end{lemma}

\subsection{Convex Functions}
\label{sec:real-analysis:convex-functions}

We begin by giving an equivalent definition of the convex hull of a function.

\begin{definition}[Convex hull of a function \citep{rockafellar1997convex05}]
\label{defn:convex-env-epi}
Let $f: [0, 1]^m \to \mathbb{R}_{\ge 0}$ be a function. The \emph{epigraph} of $f$ is
the set $\epi(f) \defeq \{(\xvec, y): \xvec \in [0, 1]^m \text{ and } y \ge f(\xvec)\}$.
Let $U_f(\xvec) \defeq \{y \in \mathbb{R}: (\xvec, y) \in \conv(\epi(f))\}$
and $\fhat(\xvec) \defeq \inf(U_f(\xvec))$ for all $\xvec \in [0, 1]^m$.
Then $\fhat: [0, 1]^m \to \mathbb{R}_{\ge 0}$ is called
the \emph{convex hull} of $f$, and is denoted by $\conv(f)$.
\end{definition}

\begin{lemma}[Theorem 5.3 in \cite{rockafellar1997convex05}]
\label{thm:conv-hull-is-conv}
Let $S \subseteq [0, 1]^m \times \mathbb{R}$ be a convex set.
Let $f(x) \defeq \inf(\{y: (x, y) \in S\})$ for all $x \in [0, 1]^m$.
Then $f: [0, 1]^m \to \mathbb{R} \cup \{-\infty, \infty\}$ is a convex function.
\end{lemma}

\begin{lemma}
\label{thm:convex-env-props-basic}
Let $f: [0, 1]^m \to \mathbb{R}_{\ge 0}$ be a function and $\fhat$ be its convex hull. Then
\begin{tightenum}
\item $\fhat(\xvec) \le f(\xvec)$ for all $\xvec \in [0, 1]^m$.
\item $\fhat$ is convex.
\item If $f$ is monotone, then $\fhat$ is also monotone.
\end{tightenum}
\end{lemma}
\begin{proof}
For any $x \in [0, 1]^m$, we have $(x, f(x)) \in \epi(f)$, so $f(x) \in U_f(x)$, so $\fhat(x) \le f(x)$.
$\fhat$ is convex by applying \cref{thm:conv-hull-is-conv} on $S = \conv(\epi(f))$.

Suppose $f$ is monotone. We want to show that $\fhat$ is also monotone.
Suppose $\xhat \preceq x^*$. Let $d \defeq x^* - \xhat$. Then $d \succeq \vecZero$.
Suppose $y^* \in U_f(x^*)$. Then $\exists k \in \mathbb{N}$, $\exists \lambda \in \Delta_k$,
and $\exists x^{(0)}, \ldots, x^{(k)} \in [0, 1]^m$ such that
\begin{align*}
x^* &= \sum_{i=1}^k \lambda_i x^{(i)},
& y^* \ge \sum_{i=1}^k \lambda_i f(x^{(i)}).
\end{align*}
Then
\begin{align*}
\xhat &= \sum_{i=1}^k \lambda_i (x^{(i)} - d),
& y^* \ge \sum_{i=1}^k \lambda_i f(x^{(i)}) \ge \sum_{i=1}^k \lambda_i f(x^{(i)} - d).
\end{align*}
Hence, $y^* \in U_f(\xhat)$. Thus, $U_f(x^*) \subseteq U_f(\xhat)$. Hence,
\[ \fhat(x^*) = \inf(U_f(x^*)) \ge \inf(U_f(\xhat)) = f(\xhat). \]
Hence, $\fhat$ is monotone.
\end{proof}

\begin{lemma}
\label{thm:mlc-epi-closed}
If $f: [0, 1]^m \to \mathbb{R}$ is MLC, then $\epi(f)$ is closed.
\end{lemma}
\begin{proof}
Let $((x^{(i)}, y^{(i)}))_{i \in \mathbb{N}}$ be a sequence in $\epi(f)$ that converges to $(x^*, y^*)$.
We must show that $(x^*, y^*) \in \epi(f)$.
Equivalently, given $y^{(i)} \ge f(x^{(i)})$ for all $i \in \mathbb{N}$,
and we must show that $y^* \ge f(x^*)$.

Let $X \defeq (x^{(i)})_{i \in \mathbb{N}}$ and $Y \defeq (y^{(i)})_{i \in \mathbb{N}}$.
By \cref{thm:convergence-decomp}, $X$ converges to $x^*$ and $Y$ converges to $y^*$.
Construct the sequence $\Xhat \defeq (\xhat^{(i)})_{i \in \mathbb{N}}$ as in \cref{thm:seq-monotonization}.
Then $\Xhat$ also converges to $x^*$.
Let $\Yhat \defeq (f(\xhat^{(i)})_{i \in \mathbb{N}}$.
Then $\Yhat$ converges to $f(x^*)$.
Moreover, $f(\xhat^{(i)}) \le f(x^{(i)}) \le y^{(i)}$.
Hence, $Y - \Yhat \defeq (y^{(i)} - f(\xhat^{(i)}))_{i \in \mathbb{N}}$
is a sequence of non-negative numbers that converges to $y^* - f(x^*)$.
Hence, $y^* \ge f(x^*)$. Thus, $(x^*, y^*) \in \epi(f)$, so $\epi(f)$ is closed.
\end{proof}

\begin{lemma}
\label{thm:mlc-conv-epi-comm}
Let $f: [0, 1]^m \to \mathbb{R}$ be an MLC function.
Let $H \defeq [0, 1]^m \times [f(\vecZero), f(\vecOne)]$.
Then $\conv(\epi(f)) \cap H = \conv(\epi(f) \cap H)$.
\end{lemma}
\begin{proof}
By \cref{thm:convex-hull-cap}, we get $\conv(\epi(f)) \cap H \supseteq \conv(\epi(f) \cap H)$,
so we need to show that $\conv(\epi(f)) \cap H \subseteq \conv(\epi(f) \cap H)$.
Let $(\xhat, \yhat) \in \conv(\epi(f)) \cap H$.
Equivalently, $(\xhat, \yhat) \in \conv(\epi(f))$ and $\yhat \le f(\vecOne)$.
Suppose $\yhat \ge f(\xhat)$. Then $(\xhat, \yhat) \in \epi(f) \cap H$.
So now assume $\yhat < f(\xhat) \le f(\vecOne)$.

Since $(\xhat, \yhat) \in \conv(\epi(f))$,
$\exists k \in \mathbb{N}$, $\exists \lambda \in \Delta_k$,
$\exists (x^{(1)}, y^{(1)}), \ldots, (x^{(k)}, y^{(k)}) \in \epi(f)$ such that
\begin{align*}
\xhat &= \sum_{i=1}^k \lambda_ix^{(i)},
& \yhat &= \sum_{i=1}^k \lambda_iy^{(i)}
\end{align*}
Assume \wLoG{} that $\lambda_i > 0$ for all $i \in [k]$.
If $y^{(i)} \le f(\vecOne)$ for all $i \in [k]$,
then $(\xhat, \yhat) \in \conv(\epi(f) \cap H)$, and we are done.
Otherwise, assume \wLoG{} that for some $k' \in [1, k-1]$,
$(x^{(i)}, y^{(i)}) \in T$ for all $i \in [k] \setminus [k - k']$.
Since $T$ is closed, we can replace these points by the single point
\[ \left(\frac{1}{L}\sum_{i=k-k'+1}^k \lambda_ix^{(i)},\;
\frac{1}{L}\sum_{i=k-k'+1}^k \lambda_iy^{(i)}\right), \]
where $L \defeq \frac{1}{L}\sum_{i=k-k'+1}^k \lambda_i$.
Hence, assume \wLoG{} that $k' = 1$.
Hence, $y^{(i)} < f(\vecOne)$ for all $i \in [k-1]$ and $y^{(k)} \ge f(\vecOne)$.

Let $z^{(i)} \defeq (x^{(i)}, y^{(i)})$ for all $i \in [k]$.
Then $z^{(i)} \in \epi(f) \cap H$ for $i \in [k-1]$ and $z^{(k)} \in T$.
Let
\[ z^{(0)} \defeq \frac{1}{1-\lambda_k}\sum_{i=1}^{k-1} \lambda_iz^{(i)}. \]
Let $\zhat \defeq (\xhat, \yhat)$.
Then $\zhat = (1-\lambda_k)z^{(0)} + \lambda_kz^{(k)}$.

Define
\[ \alpha \defeq \frac{f(\vecOne) - \yhat}{y^{(k)} - \yhat}. \]
Then $\alpha \in (0, 1]$.
Let $\xtild \defeq (1-\alpha)\xhat + \alpha x^{(k)}$,
$\ytild \defeq (1-\alpha)\yhat + \alpha y^{(k)} = f(\vecOne)$,
and $\ztild \defeq (\xtild, \ytild)$. Then $\ztild \in \epi(f) \cap H$.

\begin{align*}
& \zhat = (1-\lambda_k)z^{(0)} + \lambda_kz^{(k)}
\\ &\implies 0 = (1-\lambda_k)(z^{(0)}-\zhat) + \lambda_k(z^{(k)} - \zhat)
\\ &\implies 0 = (1-\lambda_k)(z^{(0)}-\zhat) + (\lambda_k/\alpha)(\ztild - \zhat)
\\ &\implies 0 = \frac{\alpha(1-\lambda_k)}{\alpha(1-\lambda_k)+\lambda_k}(z^{(0)}-\zhat)
    + \frac{\lambda_k}{\alpha(1-\lambda_k)+\lambda_k}(\ztild - \zhat)
\\ &\implies \zhat = \frac{\alpha(1-\lambda_k)}{\alpha(1-\lambda_k)+\lambda_k}z^{(0)}
    + \frac{\lambda_k}{\alpha(1-\lambda_k)+\lambda_k}\ztild.
\end{align*}
Hence, $\zhat$ is a convex combination of points in $\epi(f) \cap H$.
Hence, $\conv(\epi(f)) \cap H = \conv(\epi(f) \cap H)$.
\end{proof}

\begin{lemma}
\label{thm:mlc-conv-epi-closed}
If $f: [0, 1]^m \to \mathbb{R}$ is MLC, then $\conv(\epi(f))$ is closed.
\end{lemma}
\begin{proof}
We would like to use \cref{thm:mlc-epi-closed,thm:conv-closed} to show that $\conv(\epi(f))$ is closed.
Unfortunately, \cref{thm:conv-closed} only works for bounded sets, so we need a workaround.
Let $H \defeq [0, 1]^m \times [f(\vecZero), f(\vecOne)]$
and $T \defeq [0, 1]^m \times [f(\vecOne), \infty)$.
For all $(x, y) \in T$, we have $y \ge f(\vecOne) \ge f(x)$, so $(x, y) \in \epi(f)$.
For all $(x, y) \in \epi(f)$, we have $y \ge f(x) \ge f(\vecZero)$, so $(x, y) \in H \cup T$.
Hence, $T \subseteq \epi(f) \subseteq H \cup T$.
Thus, $\epi(f)$ is the disjoint union of $\epi(f) \cap H$ and $T$.

Since $\epi(f) \subseteq \conv(\epi(f))$, we get $T \subseteq \conv(\epi(f))$.
It is also easy to see that $\conv(\epi(f)) \subseteq H \cup T$.
Hence, $\conv(\epi(f))$ is the disjoint union of $\conv(\epi(f)) \cap H$ and $T$.
$T$ is closed, and the union of closed sets is closed.
Hence, to show that $\conv(\epi(f))$ is closed, we need to show that $\conv(\epi(f)) \cap H$ is closed.

By \cref{thm:mlc-conv-epi-comm}, we get $\conv(\epi(f)) \cap H = \conv(\epi(f) \cap H)$.
By \cref{thm:mlc-epi-closed}, $\epi(f)$ is closed. Since $H$ is closed,
and the intersection of closed sets is closed, we get that $\epi(f) \cap H$ is also closed.
$\epi(f) \cap H$ is also bounded.
Hence, $\conv(\epi(f) \cap H)$ is closed by \cref{thm:conv-closed}.
Thus, $\conv(\epi(f))$ is closed.
\end{proof}

\begin{lemma}
\label{thm:Ufx-has-min}
Let $f: [0, 1]^m \to \mathbb{R}_{\ge 0}$ be MLC. Then for all $\xvec \in [0, 1]^m$,
$U_f(\xvec)$ has a minimum element, i.e., $\inf(U_f(\xvec)) \in U_f(\xvec)$.
\end{lemma}
\begin{proof}
Let $\xhat \in [0, 1]^m$.
$U_f(\xhat) = \conv(\epi(f)) \cap (\{\xhat\} \times \mathbb{R})$.
$\{\xhat\} \times \mathbb{R}$ is closed,
and by \cref{thm:mlc-conv-epi-closed}, $\conv(\epi(f))$ is closed.
Since the intersection of closed sets is closed,
we get that $U_f(\xhat)$ is closed.
Thus, $\inf(U_f(\xhat)) \in U_f(\xhat)$.
\end{proof}

\begin{lemma}
\label{thm:convex-ucont}
If $f: [0, 1]^m \to \mathbb{R}$ is monotone and convex,
then $f^+(\xhat) = f(\xhat)$ for all $\xhat \in [0, 1]^m$.
\end{lemma}
\begin{proof}
Pick any $\xhat \in [0, 1]^m$.
If $\xhat = \vecOne$, then $\xhat^{[+\delta]} = \xhat$ for all $\delta > 0$,
so $f^+(\xhat) = f(\xhat)$ and we are done. So now assume $\xhat \neq \vecOne$.

Let $J \defeq \{j \in [m]: \xhat_j < 1\}$. Then $J \neq \emptyset$.
Let $\mu \defeq \min_{j \in J} (1-\xhat_j)$. Then $\mu > 0$.
For all $j \in [m]$, let $d_j$ be 1 if $j \in J$ and 0 otherwise.
Then $\xhat + \mu d \in [0, 1]^m$.

Let $0 < \delta < \mu$. Then $\xhat^{[+\delta]} = \xhat + \delta d$.
Moreover, $\xhat + \delta d$ is a convex combination of $\xhat$ and $\xhat + \mu d$.
Specifically, we have
\[ \xhat + \delta d = \left(1 - \frac{\delta}{\mu}\right)\xhat + \left(\frac{\delta}{\mu}\right)(\xhat + \mu d). \]
Since $f$ is convex, we get
\begin{align*}
f^+(\xhat) &\le f(\xhat^{[+\delta]}) = f(\xhat + \delta d)
\\ &\le \left(1 - \frac{\delta}{\mu}\right)f(\xhat) + \left(\frac{\delta}{\mu}\right)f(\xhat + \mu d)
\\ &= f(\xhat) + \frac{\delta}{\mu}(f(\xhat + \mu d) - f(\xhat)).
\end{align*}
Since this holds for all $\delta > 0$, we get $f^+(\xhat) \le f(\xhat)$.
Hence, $f^+(\xhat) = f(\xhat)$.
\end{proof}

\begin{lemma}
\label{thm:convex-env-lcont}
Let $f: [0, 1]^m \to \mathbb{R}_{\ge 0}$ be an MLC function. Let $\fhat \defeq \conv(f)$.
Then $\fhat^-(\xhat) = \fhat(\xhat)$ for all $\xhat \in [0, 1]^m$.
\end{lemma}
\begin{proof}
Let $C \defeq \conv(\epi(f))$.
For all $i \in \mathbb{N}$, define $x^{(i)} \defeq \xhat^{[-2^{-i}]}$ and $y^{(i)} \defeq \fhat(x^{(i)})$.
Let $X \defeq (x^{(i)})_{i \in \mathbb{N}}$ and $Y \defeq (y^{(i)})_{i \in \mathbb{N}}$.
Then $X$ converges to $\xhat$ and $\sup(Y) = f^-(\xhat)$.
Let $Z \defeq ((x^{(i)}, y^{(i)}))_{i \in \mathbb{N}}$.
By \cref{thm:convergence-decomp}, $Z$ converges to $(\xhat, f^-(\xhat))$.

By \cref{thm:Ufx-has-min}, $(x^{(i)}, \fhat(x^{(i)})) \in C$.
Hence, each element of $Z$ lies in $C$.
By \cref{thm:mlc-epi-closed}, $C$ is closed.
Hence, $(\xhat, f^-(\xhat))$ also lies in $C$.
Thus, $\fhat(\xhat) \le f^-(\xhat)$, and so $\fhat(\xhat) = f^-(\xhat)$.
\end{proof}

\thmConvexEnvProps*
\label{prf:convex-env-props}
\begin{proof}
The first three claims follow from \cref{thm:convex-env-props-basic}.
$\fhat^-(x) = \fhat(x)$ for all $x \in [0, 1]^m$ by \cref{thm:convex-env-lcont}.
By \cref{thm:convex-env-props-basic}, $\fhat$ is monotone and convex,
so $\fhat^+(x) = \fhat(x)$ for all $x \in [0, 1]^m$ by \cref{thm:convex-ucont}.
Hence, $\fhat$ is continuous.
\end{proof}

\subsection{Left Derivative of Convex Functions}
\label{sec:real-analysis:ld}

\begin{definition}[Left derivative]
\label{defn:convex-ld}
Let $f: [0, 1] \to \mathbb{R}$ be a convex function.
The \emph{left derivative} of $f$ at $\xhat \in (0, 1]$ is defined as
\[ f'(\xhat) \defeq \sup_{\delta \in (0, \xhat]} \frac{f(\xhat) - f(\xhat - \delta)}{\delta}. \]
Moreover, we define $f'(0)$ to be $-\infty$.
\end{definition}

For convex functions, one can show using \cref{thm:convex-rates}
that \cref{defn:convex-ld} is equivalent to the more well-known definition
\[ f'(\xhat) \defeq \lim_{\delta \to 0^-} \frac{f(\xhat) - f(\xhat - \delta)}{\delta}. \]

\begin{lemma}
\label{thm:convex-rates}
Let $f: [0, 1] \to \mathbb{R}$ be a convex function. Let $0 \le a < \xhat < b \le 1$. Then
\[ \frac{f(\xhat) - f(a)}{\xhat - a} \le \frac{f(b) - f(a)}{b - a} \le \frac{f(b) - f(\xhat)}{b - \xhat}. \]
\end{lemma}
\begin{proof}
$\xhat$ is a convex combination of $a$ and $b$, since
\[ \xhat = \left(\frac{b-\xhat}{b-a}\right)a + \left(\frac{\xhat-a}{b-a}\right)b. \]
Since $f$ is convex, we get
\[ f(\xhat) \le \left(\frac{b-\xhat}{b-a}\right)f(a) + \left(\frac{\xhat-a}{b-a}\right)f(b). \]
Hence,
\begin{align*}
f(\xhat) - f(a) &\le \left(\frac{\xhat-a}{b-a}\right)(f(b)-f(a)),
& f(b) - f(\xhat) &\ge \left(\frac{b-\xhat}{b-a}\right)(f(b)-f(a)).
\qedhere
\end{align*}
\end{proof}

\begin{lemma}
\label{thm:convex-ld-rates}
Let $f: [0, 1] \to \mathbb{R}$ be a convex function. Let $0 \le a < b \le 1$. Then
\[ f'(a) \le \frac{f(b) - f(a)}{b-a} \le f'(b). \]
\end{lemma}
\begin{proof}
By \cref{thm:convex-rates}, for all $\delta \in (0, a]$, we have
\[ \frac{f(a)-f(a-\delta)}{\delta} \le \frac{f(b)-f(a)}{b-a}. \]
Hence, $f'(a) \le (f(b)-f(a))/(b-a)$.

By \cref{thm:convex-rates}, for all $\delta \in (0, b-a]$, we have
\[ \frac{f(b)-f(a)}{b-a} \le \frac{f(b)-f(b-\delta)}{\delta}. \]
Hence, $(f(b)-f(a))/(b-a) \le f'(b)$.
\end{proof}

\begin{lemma}
\label{thm:convex-ld-mlc}
Let $f: [0, 1] \to \mathbb{R}$ be a lower-continuous convex function. Then $f'$ is MLC.
Moreover, if $f$ is monotone, then $f'(x) \ge 0$ for all $x \in (0, 1]$.
\end{lemma}
\begin{proof}
Let $g = f'$. $g$ is monotone by \cref{thm:convex-ld-rates}.
We will now show that $g^-(\xhat) = g(\xhat)$ for all $\xhat \in [0, 1]$.
This is trivially true for $\xhat = 0$, so now assume $\xhat > 0$.

Suppose $g(\xhat) \neq g^-(\xhat)$. Then $g^-(\xhat)$ is finite.
Pick any $\delta \in (0, \xhat)$.
By \cref{thm:convex-ld-rates}, for any $x \in (\xhat-\delta, \xhat)$, we get
\[ \frac{f(x) - f(\xhat-\delta)}{x - \xhat + \delta} \le f'(x) = g(x) \le g^-(\xhat). \]
Hence, $f(x) \le f(\xhat-\delta) + g^-(\xhat)(x - \xhat + \delta) \le f(\xhat-\delta) + \delta g^-(\xhat)$.
Since $f$ is lower-continuous, we get
\[ f(\xhat) = f^-(\xhat) \le \sup_{\delta' \in (0, \delta)} \inf_{\xhat-\delta' \le x \le \xhat} f(x)
    \le f(\xhat-\delta) + \delta g^-(\xhat). \]
Hence,
\[ \frac{f(\xhat) - f(\xhat - \delta)}{\delta} \le g^-(\xhat). \]
Since this holds for arbitrarily small $\delta$, we get
$g(\xhat) = f'(\xhat) \le g^-(\xhat)$.
Hence, $g$ is left-continuous at $\xhat$.

If $f$ is monotone, then
\[ f'(\xhat) = \sup_{\delta \in (0, \xhat]} \frac{f(\xhat) - f(\xhat - \delta)}{\delta} \ge 0.
\qedhere \]
\end{proof}

\begin{lemma}
\label{thm:convex-ld-sum}
Let $f, g: [0, 1] \to \mathbb{R}$ be convex functions.
Then $f+g$ is also convex and $(f+g)' = f'+g'$.
\end{lemma}
\begin{proof}
Let $\alpha, x, y \in [0, 1]$. Then
\begin{align*}
& (f+g)((1-\alpha)x + \alpha y)
\\ &= f((1-\alpha)x + \alpha y) + g((1-\alpha)x + \alpha y)
\\ &\le ((1-\alpha)f(x) + \alpha f(y)) + ((1-\alpha)g(x) + \alpha g(y))
\\ &= (1-\alpha)(f+g)(x) + \alpha(f+g)(y).
\end{align*}
Hence, $f+g$ is convex.

$(f+g)'(0) = -\infty$ and $f'(0) + g'(0) = (-\infty) + (-\infty) = -\infty$.
For any $x \in (0, 1]$ and $\delta \in (0, x]$, we have
\[ \frac{(f+g)(x) - (f+g)(x-\delta)}{\delta}
    = \frac{f(x) - f(x-\delta)}{\delta} + \frac{g(x) - g(x-\delta)}{\delta}. \]
By \cref{thm:sup-sum}, we get $(f+g)'(x) \le f'(x) + g'(x)$.

If $f'(x) = \infty$ or $g'(x) = \infty$, then at least one of
\[ \frac{f(x) - f(x-\delta)}{\delta} \quad\text{or}\quad \frac{g(x) - g(x-\delta)}{\delta} \]
can be made arbitrarily large for some $\delta$, and so, we get that $(f+g)'(x) = \infty$.
Now assume $f'(x)$ and $g'(x)$ are finite.

Pick any $\eps > 0$. Then $\exists \delta > 0$ such that
\[ \frac{f(x) - f(x-\delta)}{\delta} > f'(x) - \eps
    \quad\text{and}\quad \frac{g(x) - g(x-\delta)}{\delta} > g'(x) - \eps. \]
Then
\[ (f+g)'(x) \ge \frac{(f+g)'(x) - (f+g)'(x-\delta)}{\delta}
    > f'(x) + g'(x) - 2\eps. \]
Hence, $(f+g)'(x) \ge f'(x) + g'(x)$.
Thus, we get $(f+g)' = f' + g'$.
\end{proof}

\begin{lemma}
\label{thm:convex-ld-min}
Let $f: [0, 1] \to \mathbb{R}$ be a lower-continuous convex function.
Let $x^* \defeq \sup(\{x \in [0, 1]: f'(x) \le 0\})$.
Then $x^*$ is the rightmost global minimum of $f$, i.e.,
$f(x^*) \le f(x)$ for all $x \in [0, x^*]$
and $f(x^*) < f(x)$ for all $x \in (x^*, 1]$.
\end{lemma}
\begin{proof}
By \cref{thm:convex-ld-mlc}, $f'$ is also MLC,
and by \cref{thm:mlc-feas-compact}, $\{x \in [0, 1]: g(x) \le \beta\}$ is closed.
Thus, $f'(x^*) \le 0$.

Let $x < x^*$. Then by \cref{thm:convex-ld-rates}, we get
\[ \frac{f(x^*) - f(x)}{x^* - x} \le f'(x^*) \le 0, \]
so $f(x^*) \le f(x)$.

Let $x^* < x^{(1)} < x^{(2)} \le 1$.
By definition of $x^*$, we get that $f'(x^{(1)}) > 0$.
Moreover, using \cref{thm:convex-ld-rates}, we get
\[ 0 < f'(x^{(1)}) \le \frac{f(x^{(2)}) - f(x^{(1)})}{x^{(2)} - x^{(1)}} \implies f(x^{(1)}) < f(x^{(2)}). \]
Thus, $f$ is monotone in the interval $(x^*, 1]$.
Let $\alpha \defeq \inf_{x \in (x^*, 1]} f(x)$. Then
\begin{align*}
f^-(x^*) &= \sup_{\delta > 0} \inf_{x \in N_\delta(x^*)} f(x)
\\ &= \sup_{\delta > 0} \min\left(\inf_{x \in [\max(0, x^* - \delta], x^*]} f(x),
    \inf_{x \in (x^*, \min(1, x^* + \delta)]} f(x)\right)
    \tag{by \cref{thm:sup-inf-union}}
\\ &= \min(f(x^*), \alpha).
\end{align*}
If $\alpha < f(x^*)$, then we get that $f^-(x^*) = \alpha < f(x^*)$,
which contradicts the fact that $f$ is lower continuous.
Thus, $\alpha \ge f(x^*)$, and hence, $f(x) \ge f(x^*)$ for all $x \in (x^*, 1]$.

Suppose $f(\xhat) = f(x^*)$ for some $\xhat \in (x^*, 1]$.
Pick any $x \in (x^*, \xhat)$. Then $f(x) = f(x^*)$.
Moreover, using \cref{thm:convex-ld-rates}, we get
\[ 0 = \frac{f(x) - f(x^*)}{x - x^*} \le f'(x) \le \frac{f(\xhat) - f(x)}{\xhat - x} = 0. \]
Thus, $f'(x) = 0$, which contradicts the definition of $x^*$.
Hence, $f(\xhat) > f(x^*)$ for all $\xhat \in (x^*, 1]$.
\end{proof}

\subsection{Separable Functions}
\label{sec:real-analysis:sep}

\thmSepProps*
\label{prf:sep-props}
\begin{proof}
\textbf{Monotonicity}:
Let $\xvec^{(1)}, \xvec^{(2)} \in \mathbb{R}^m$ such that $\xvec^{(1)} \preceq \xvec^{(2)}$.
Let $\yvec^{(1)}, \yvec^{(2)} \in \mathbb{R}^n$ such that $\yvec^{(1)} \preceq \yvec^{(2)}$.
\begin{itemize}
\item If $f$ and $g$ are monotone, then
    $h(\xvec^{(1)}, \yvec^{(1)}) = f(\xvec^{(1)}) + g(\yvec^{(1)})
    \le f(\xvec^{(2)}) + g(\yvec^{(2)}) = h(\xvec^{(2)}, \yvec^{(2)})$,
    so $h$ is monotone.
\item Suppose $h$ is monotone. Then
    $f(\xvec^{(1)}) = h(\xvec^{(1)}, \vecZero) - g(\vecZero)
    \le h(\xvec^{(2)}, \vecZero) - g(\vecZero) = f(\xvec^{(2)})$.
    Hence, $f$ is monotone. Similarly, $g$ is monotone.
\end{itemize}

\textbf{Continuity}:
Let $\xvechat \in [0, 1]^m$ and $\yvechat \in [0, 1]^n$.
Recall the definition of continuity from \cref{defn:continuity}.
For any $\delta > 0$, we have $N_{\delta}((\xvechat, \yvechat)) = N_{\delta}(\xvechat) \times N_{\delta}(\yvechat)$.
So,
\begin{align*}
h^-(\xvechat, \yvechat) &= \sup_{\delta > 0} \inf_{(\xvec, \yvec) \in N_{\delta}((\xvechat, \yvechat))} h(\xvec, \yvec)
\\ &= \sup_{\delta > 0} \left(\inf_{\xvec \in N_{\delta}(\xvechat)} f(\xvec)
    + \inf_{\yvec \in N_{\delta}(\yvechat)} g(\yvec)\right)
\\ &\le f^-(\xvechat) + g^-(\yvechat)  \tag{by \cref{thm:sup-sum}}
\end{align*}
Let $\eps > 0$. Then for some $\delta_1 > 0$, we have
\[ \inf_{\xvec \in N_{\delta_1}(\xvechat)} f(\xvec) > f^-(\xvechat) - \eps, \]
and for some $\delta_2 > 0$, we have
\[ \inf_{\yvec \in N_{\delta_1}(\yvechat)} g(\yvec) > g^-(\yvechat) - \eps. \]
Let $\deltahat \defeq \min(\delta_1, \delta_2)$. Then
\begin{align*}
h^-(\xvechat, \yvechat) &= \sup_{\delta > 0} \left(\inf_{\xvec \in N_{\delta}(\xvechat)} f(\xvec)
    + \inf_{\yvec \in N_{\delta}(\yvechat)} g(\yvec)\right)
\\ &\ge \left(\inf_{\xvec \in N_{\deltahat}(\xvechat)} f(\xvec)
    + \inf_{\yvec \in N_{\deltahat}(\yvechat)} g(\yvec)\right)
\\ &\ge \left(\inf_{\xvec \in N_{\delta_1}(\xvechat)} f(\xvec)
    + \inf_{\yvec \in N_{\delta_2}(\yvechat)} g(\yvec)\right)
\\ &> f^-(\xvechat) + g^-(\yvechat) - 2\eps.
\end{align*}
Since this holds for all $\eps > 0$, we get $h^-(\xvechat, \yvechat) = f^-(\xvechat) + g^-(\yvechat)$.

\begin{itemize}
\item If $f$ and $g$ are lower continuous, then
    $h^-(\xvechat, \yvechat) = f^-(\xvechat) + g^-(\yvechat) = f(\xvechat) + g(\yvechat) = h(\xvechat, \yvechat)$,
    so $h$ is lower continuous too.
\item If $h$ is lower continuous, then
    \[ 0 = h(\xvechat, \yvechat) - h^-(\xvechat, \yvechat)
        = (f(\xvechat) - f^-(\xvechat)) + (g(\yvechat) - g^-(\yvechat)). \]
    Since $f(\xvechat) \ge f^-(\xvechat)$ and $g(\yvechat) \ge g^-(\yvechat)$,
    we get that $f(\xvechat) - f^-(\xvechat) = g(\yvechat) - g^-(\yvechat) = 0$,
    so $f$ and $g$ are also lower continuous.
\end{itemize}

We can similarly prove that $f$ and $g$ are upper continuous iff $h$ is upper continuous.

\textbf{Convexity}:
Let $\xvec^{(1)}, \xvec^{(2)} \in \mathbb{R}^m$ such that $\xvec^{(1)} \preceq \xvec^{(2)}$.
Let $\yvec^{(1)}, \yvec^{(2)} \in \mathbb{R}^n$ such that $\yvec^{(1)} \preceq \yvec^{(2)}$.
\begin{itemize}
\item If $f$ and $g$ are convex, then
    \begin{align*}
    & h((1-\alpha)(\xvec^{(1)}, \yvec^{(1)}) + \alpha(\xvec^{(2)}, \yvec^{(2)}))
    \\ &= f((1-\alpha)\xvec^{(1)} + \alpha\xvec^{(2)}) + g((1-\alpha)\yvec^{(1)} + \alpha\yvec^{(2)})
    \\ &\le ((1-\alpha)f(\xvec^{(1)}) + \alpha f(\xvec^{(2)})) + ((1-\alpha)g(\yvec^{(1)}) + \alpha g(\yvec^{(2)}))
    \\ &= (1-\alpha)h(\xvec^{(1)}, \yvec^{(1)}) + \alpha h(\xvec^{(2)}, \yvec^{(2)}),
    \end{align*}
    so $h$ is convex.
\item Suppose $h$ is convex. Then
    \begin{align*}
    & f((1-\alpha)\xvec^{(1)} + \alpha\xvec^{(2)})
    \\ &= h((1-\alpha)\xvec^{(1)} + \alpha\xvec^{(2)}, \vecZero) - g(\vecZero)
    \\ &\le (1-\alpha)h(\xvec^{(1)}, \vecZero) + \alpha h(\xvec^{(2)}, \vecZero) - g(\vecZero)
    \\ &= (1-\alpha)f(\xvec^{(1)}) + \alpha f(\xvec^{(2)}).
    \end{align*}
    Hence, $f$ is convex. Similarly, $g$ is convex.
\qedhere
\end{itemize}
\end{proof}

\thmSepOpt*
\label{prf:sep-opt}
\begin{proof}
Let $\xvec^* \in X^*$ and $\yvec^* \in Y^*$. Let $(\xvechat, \yvechat) \in Z^*$. Then
$h(\xvec^*, \yvec^*) = f(\xvec^*) + g(\yvec^*) \ge f(\xvechat) + g(\yvechat)
= h(\xvechat, \yvechat) \ge h(\xvec^*, \yvec^*)$.
Hence, all inequalities here hold with equality.
Thus, $h(\xvec^*, \yvec^*) = h(\xvechat, \yvechat)$,
$f(\xvec^*) = f(\xvechat)$, and $g(\yvec^*) = g(\yvechat)$.
Hence, $\xvechat \in X^*$, $\yvechat \in Y^*$, and $(\xvec^*, \yvec^*) \in Z^*$.
Thus, $Z^* = X^* \times Y^*$.
\end{proof}

\end{document}

%% file: dm-monopoly-final.bbl
\newcommand{\etalchar}[1]{$^{#1}$}

%% file: dm-monopoly-final.bbl
\begin{thebibliography}{FMMO24}

\bibitem[ACGM25]{akrami2025theoretical}
Hannaneh Akrami, Bhaskar~Ray Chaudhury, Jugal Garg, and Aniket Murhekar.
\newblock On the theoretical foundations of data exchange economies.
\newblock In {\em Conf.\ Economics and Computation (EC)}, pages 444--444, 2025.
\newblock \href {https://doi.org/10.1145/3736252.3742566}
  {\path{doi:10.1145/3736252.3742566}}.

\bibitem[{Acu}22]{AcumenDM}
{Acumen Research}.
\newblock Big data market size: Global industry, share, analysis, trends and
  forecast 2022 - 2030, 2022.
\newblock URL:
  \url{https://www.acumenresearchandconsulting.com/big-data-market}.

\bibitem[ADHR24]{AgarwalDHR20}
Anish Agarwal, Munther~A. Dahleh, Thibaut Horel, and Maryann Rui.
\newblock Towards data auctions with externalities.
\newblock {\em Games Econ. Behav.}, 148:323--356, 2024.
\newblock \href {https://doi.org/10.1016/j.geb.2024.09.008}
  {\path{doi:10.1016/j.geb.2024.09.008}}.

\bibitem[ADS19]{AgarwalDS19}
Anish Agarwal, Munther~A. Dahleh, and Tuhin Sarkar.
\newblock A marketplace for data: An algorithmic solution.
\newblock In {\em Conf.\ Economics and Computation (EC)}, pages 701--726.
  {ACM}, 2019.
\newblock \href {https://doi.org/10.1145/3328526.3329589}
  {\path{doi:10.1145/3328526.3329589}}.

\bibitem[AFM{\etalchar{+}}23]{acemoglu2023good}
Daron Acemoglu, Alireza Fallah, Ali Makhdoumi, Azarakhsh Malekian, and Asuman
  Ozdaglar.
\newblock How good are privacy guarantees? platform architecture and violation
  of user privacy.
\newblock Technical report, National Bureau of Economic Research, 2023.

\bibitem[AP86]{admati1986monopolistic}
Anat~R Admati and Paul Pfleiderer.
\newblock A monopolistic market for information.
\newblock {\em Journal of Economic Theory}, 39(2):400--438, 1986.
\newblock \href {https://doi.org/10.1016/0022-0531(86)90052-9}
  {\path{doi:10.1016/0022-0531(86)90052-9}}.

\bibitem[AP90]{admati1990direct}
Anat~R Admati and Paul Pfleiderer.
\newblock Direct and indirect sale of information.
\newblock {\em Econometrica: Journal of the Econometric Society}, pages
  901--928, 1990.
\newblock \href {https://doi.org/10.2307/2938355} {\path{doi:10.2307/2938355}}.

\bibitem[BBG22]{bergemann2022economics}
Dirk Bergemann, Alessandro Bonatti, and Tan Gan.
\newblock The economics of social data.
\newblock {\em The RAND Journal of Economics}, 53(2):263--296, 2022.
\newblock \href {https://doi.org/10.1111/1756-2171.12407}
  {\path{doi:10.1111/1756-2171.12407}}.

\bibitem[BBS18]{bergemann2018design}
Dirk Bergemann, Alessandro Bonatti, and Alex Smolin.
\newblock The design and price of information.
\newblock {\em American economic review}, 108(1):1--48, 2018.
\newblock \href {https://doi.org/10.1257/aer.20161079}
  {\path{doi:10.1257/aer.20161079}}.

\bibitem[BGI{\etalchar{+}}24]{BhaskaraGIKMS24}
Aditya Bhaskara, Sreenivas Gollapudi, Sungjin Im, Kostas Kollias, Kamesh
  Munagala, and Govind~S. Sankar.
\newblock Data exchange markets via utility balancing.
\newblock In {\em World Wide Web Conf.\ (WWW)}, pages 57--65. {ACM}, 2024.
\newblock \href {https://doi.org/10.1145/3589334.3645364}
  {\path{doi:10.1145/3589334.3645364}}.

\bibitem[BJS10]{bazaraa2010linear}
M.S. Bazaraa, J.J. Jarvis, and H.D. Sherali.
\newblock {\em Linear Programming and Network Flows}.
\newblock Wiley, 2010.

\bibitem[BKL12]{BabaioffKP12}
Moshe Babaioff, Robert Kleinberg, and Renato~Paes Leme.
\newblock Optimal mechanisms for selling information.
\newblock In {\em Conf.\ Electronic Commerce (EC)}, page 92–109, 2012.
\newblock \href {https://doi.org/10.1145/2229012.2229024}
  {\path{doi:10.1145/2229012.2229024}}.

\bibitem[BT97]{bertsimas1997introduction}
Dimitris Bertsimas and John~N. Tsitsiklis.
\newblock {\em Introduction to Linear Optimization}.
\newblock Athena Scientific, 1997.
\newblock URL: \url{http://www.athenasc.com/linoptbook.html}.

\bibitem[BV25]{baley2025data}
Isaac Baley and Laura~L. Veldkamp.
\newblock {\em The Data Economy: Tools and Applications}.
\newblock Princeton University Press, 2025.
\newblock \href {https://doi.org/10.1515/9780691256740}
  {\path{doi:10.1515/9780691256740}}.

\bibitem[CEP{\etalchar{+}}23]{cummings2023optimal}
Rachel Cummings, Hadi Elzayn, Emmanouil Pountourakis, Vasilis Gkatzelis, and
  Juba Ziani.
\newblock Optimal data acquisition with privacy-aware agents.
\newblock In {\em IEEE Conference on Secure and Trustworthy Machine Learning
  (SaTML)}, pages 210--224, 2023.
\newblock \href {https://doi.org/10.1109/SaTML54575.2023.00023}
  {\path{doi:10.1109/SaTML54575.2023.00023}}.

\bibitem[CGMS26]{ChaudhuryGMS26}
Bhaskar~Ray Chaudhury, Jugal Garg, Aniket Murhekar, and Jiaxin Song.
\newblock Data pricing via competitive equilibrium.
\newblock In {\em Proceedings of the {ACM} on Web Conference (WWW)}, 2026.

\bibitem[CGSS26a]{ChaudhuryGSS2026}
Bhaskar~Ray Chaudhury, Jugal Garg, Eklavya Sharma, and Jiaxin Song.
\newblock Equilibrium pricing in oligopolistic data markets.
\newblock In {\em Proceedings of the International Conference on Machine
  Learning (ICML)}, 2026.
\newblock Spotlight.
\newblock URL: \url{http://jugal.ise.illinois.edu/papers/icml26.pdf}.

\bibitem[CGSS26b]{chaudhury2026revenue}
Bhaskar~Ray Chaudhury, Jugal Garg, Eklavya Sharma, and Jiaxin Song.
\newblock Revenue-optimal pricing for budget-constrained buyers in data
  markets, 2026.
\newblock \href {https://arxiv.org/abs/2602.13897v2}
  {\path{arXiv:2602.13897v2}}.

\bibitem[CHK24]{ChenHK24}
Keran Chen, Joon~Suk Huh, and Kirthevasan Kandasamy.
\newblock Learning to price homogeneous data.
\newblock In {\em Conf.\ Adv.\ Neural Information Processing Systems
  (NeurIPS)}, 2024.

\bibitem[CIL{\etalchar{+}}18]{chen2018optimal}
Yiling Chen, Nicole Immorlica, Brendan Lucier, Vasilis Syrgkanis, and Juba
  Ziani.
\newblock Optimal data acquisition for statistical estimation.
\newblock In {\em Proceedings of the 2018 ACM Conference on Economics and
  Computation}, pages 27--44, 2018.

\bibitem[CLR{\etalchar{+}}15]{cummings2015accuracy}
Rachel Cummings, Katrina Ligett, Aaron Roth, Zhiwei~Steven Wu, and Juba Ziani.
\newblock Accuracy for sale: Aggregating data with a variance constraint.
\newblock In {\em Symp.\ Innovations in Theoret.\ Computer Science (ITCS)},
  pages 317--324, 2015.
\newblock \href {https://doi.org/10.1145/2688073.2688106}
  {\path{doi:10.1145/2688073.2688106}}.

\bibitem[CRTT22]{ChawlaRTT22}
Shuchi Chawla, Rojin Rezvan, Yifeng Teng, and Christos Tzamos.
\newblock Pricing ordered items.
\newblock In {\em Symp.\ Theory of Computing (STOC)}, pages 722--735, 2022.
\newblock \href {https://doi.org/10.1145/3519935.3520065}
  {\path{doi:10.1145/3519935.3520065}}.

\bibitem[CV21]{cai2020sell}
Yang Cai and Grigoris Velegkas.
\newblock How to sell information optimally: An algorithmic study.
\newblock In {\em Symp.\ Innovations in Theoret.\ Computer Science (ITCS)},
  2021.
\newblock \href {https://doi.org/10.4230/LIPIcs.ITCS.2021.81}
  {\path{doi:10.4230/LIPIcs.ITCS.2021.81}}.

\bibitem[DDT14]{daskalakis2014complexity}
Constantinos Daskalakis, Alan Deckelbaum, and Christos Tzamos.
\newblock The complexity of optimal mechanism design.
\newblock In {\em Symp.\ Discrete Algorithms (SODA)}, pages 1302--1318. SIAM,
  2014.
\newblock \href {https://doi.org/10.1137/1.9781611973402.96}
  {\path{doi:10.1137/1.9781611973402.96}}.

\bibitem[FMMO22]{fallah2022bridging}
Alireza Fallah, Ali Makhdoumi, Azarakhsh Malekian, and Asuman Ozdaglar.
\newblock Bridging central and local differential privacy in data acquisition
  mechanisms.
\newblock {\em Conf.\ Adv.\ Neural Information Processing Systems (NeurIPS)},
  35:21628--21639, 2022.

\bibitem[FMMO24]{fallah2024optimal}
Alireza Fallah, Ali Makhdoumi, Azarakhsh Malekian, and Asuman Ozdaglar.
\newblock Optimal and differentially private data acquisition: Central and
  local mechanisms.
\newblock {\em Operations Research}, 72(3):1105--1123, 2024.
\newblock \href {https://doi.org/10.1287/opre.2022.0014}
  {\path{doi:10.1287/opre.2022.0014}}.

\bibitem[FOT23]{fleckenstein2023review}
Mike Fleckenstein, Ali Obaidi, and Nektaria Tryfona.
\newblock A review of data valuation approaches and building and scoring a data
  valuation model.
\newblock {\em Harvard Data Science Review}, 5(1), 2023.
\newblock \href {https://doi.org/10.1162/99608f92.c18db966}
  {\path{doi:10.1162/99608f92.c18db966}}.

\bibitem[FSVV25]{farboodi2025valuing}
Maryam Farboodi, Dhruv Singal, Laura Veldkamp, and Venky Venkateswaran.
\newblock Valuing financial data.
\newblock {\em The Review of Financial Studies}, 38(3):938--980, 2025.
\newblock \href {https://doi.org/10.1093/rfs/hhae034}
  {\path{doi:10.1093/rfs/hhae034}}.

\bibitem[FV23]{farboodi2023data}
Maryam Farboodi and Laura Veldkamp.
\newblock Data and markets.
\newblock {\em Annual Review of Economics}, 15(1):23--40, 2023.
\newblock \href {https://doi.org/10.1146/annurev-economics-082322-023244}
  {\path{doi:10.1146/annurev-economics-082322-023244}}.

\bibitem[GR11]{ghosh2011selling}
Arpita Ghosh and Aaron Roth.
\newblock Selling privacy at auction.
\newblock In {\em Conf.\ Electronic Commerce (EC)}, pages 199--208, 2011.
\newblock \href {https://doi.org/10.1145/1993574.1993605}
  {\path{doi:10.1145/1993574.1993605}}.

\bibitem[HC24]{Hossain024}
Safwan Hossain and Yiling Chen.
\newblock Equilibrium of data markets with externality.
\newblock In {\em {ICML}}. {PMLR}, 2024.
\newblock URL: \url{https://proceedings.mlr.press/v235/hossain24a.html}.

\bibitem[HS16]{horner2016selling}
Johannes H{\"o}rner and Andrzej Skrzypacz.
\newblock Selling information.
\newblock {\em Journal of Political Economy}, 124(6):1515--1562, 2016.

\bibitem[Ich21]{ichihashi2021competing}
Shota Ichihashi.
\newblock Competing data intermediaries.
\newblock {\em The RAND Journal of Economics}, 52(3):515--537, 2021.
\newblock \href {https://doi.org/10.1111/1756-2171.12382}
  {\path{doi:10.1111/1756-2171.12382}}.

\bibitem[Kle04]{klemperer2004auctions}
Paul Klemperer.
\newblock {\em Auctions: Theory and Practice}.
\newblock Princeton University Press, 2004.

\bibitem[MDJM21]{mehta2021sell}
Sameer Mehta, Milind Dawande, Ganesh Janakiraman, and Vijay Mookerjee.
\newblock How to sell a data set? pricing policies for data monetization.
\newblock {\em Information Systems Research}, 32(4):1281--1297, 2021.
\newblock \href {https://doi.org/10.1287/isre.2021.1027}
  {\path{doi:10.1287/isre.2021.1027}}.

\bibitem[MYC{\etalchar{+}}23]{MurhekarYCLM23}
Aniket Murhekar, Zhuowen Yuan, Bhaskar~Ray Chaudhury, Bo~Li, and Ruta Mehta.
\newblock Incentives in federated learning: Equilibria, dynamics, and
  mechanisms for welfare maximization.
\newblock In {\em Conf.\ Adv.\ Neural Information Processing Systems
  (NeurIPS)}, 2023.

\bibitem[Mye81]{myerson1981optimal}
Roger~B Myerson.
\newblock Optimal auction design.
\newblock {\em Mathematics of operations research}, 6(1):58--73, 1981.
\newblock \href {https://doi.org/10.1287/moor.6.1.58}
  {\path{doi:10.1287/moor.6.1.58}}.

\bibitem[NVX14]{nissim2014redrawing}
Kobbi Nissim, Salil Vadhan, and David Xiao.
\newblock Redrawing the boundaries on purchasing data from privacy-sensitive
  individuals.
\newblock In {\em Proceedings of the 5th conference on Innovations in
  theoretical computer science}, pages 411--422, 2014.

\bibitem[PB97]{caliHousingDataset}
R~Kelley Pace and Ronald Barry.
\newblock Sparse spatial autoregressions.
\newblock {\em Statistics \& Probability Letters}, 33(3):291--297, 1997.
\newblock URL:
  \url{https://scikit-learn.org/stable/datasets/real_world.html#california-housing-dataset},
  \href {https://doi.org/10.1016/S0167-7152(96)00140-X}
  {\path{doi:10.1016/S0167-7152(96)00140-X}}.

\bibitem[Pei20]{pei2020survey}
Jian Pei.
\newblock A survey on data pricing: from economics to data science.
\newblock {\em IEEE Transactions on knowledge and Data Engineering},
  34(10):4586--4608, 2020.
\newblock \href {https://doi.org/10.1109/TKDE.2020.3045927}
  {\path{doi:10.1109/TKDE.2020.3045927}}.

\bibitem[Roc97a]{rockafellar1997convex05}
R.T. Rockafellar.
\newblock {\em Convex Analysis}, chapter 5: Functional Operators.
\newblock Princeton University Press, 1997.
\newblock URL: \url{https://books.google.com/books?id=1TiOka9bx3sC}.

\bibitem[Roc97b]{rockafellar1997convex17}
R.T. Rockafellar.
\newblock {\em Convex Analysis}, chapter 17: Carath{\'e}odory's Theorem.
\newblock Princeton University Press, 1997.
\newblock URL: \url{https://books.google.com/books?id=1TiOka9bx3sC}.

\bibitem[SKSC25]{song2025existence}
Jiaxin Song, Pooja Kulkarni, Parnian Shahkar, and Bhaskar~Ray Chaudhury.
\newblock On the existence and complexity of core-stable data exchanges.
\newblock {\em arXiv preprint arXiv:2509.16450}, 2025.

\bibitem[Var09]{varian2009economic}
Hal~R Varian.
\newblock Economic aspects of personal privacy.
\newblock In {\em Internet Policy and Economics: Challenges and Perspectives},
  pages 101--109. Springer, 2009.

\bibitem[Vel23]{veldkamp2023valuing}
Laura Veldkamp.
\newblock Valuing data as an asset.
\newblock {\em Review of Finance}, 27(5):1545--1562, 2023.
\newblock \href {https://doi.org/10.1093/rof/rfac073}
  {\path{doi:10.1093/rof/rfac073}}.

\end{thebibliography}
